\documentclass[twocolumn]{autart}    

\usepackage{graphicx}          

\usepackage{amsmath}
\usepackage{mathtools}
\mathtoolsset{showonlyrefs}
\usepackage{amssymb}
\usepackage{verbatim}
\usepackage{graphicx}
\usepackage{enumitem}
\usepackage{float}
\usepackage{cases}
\usepackage{xcolor}
\usepackage{ulem}
\usepackage[numbers, square, comma, sort&compress]{natbib}
\allowdisplaybreaks
\begin{document}
\definecolor{ForestGreen}{rgb}{0.0, 0.27, 0.13}
\newcommand{\R}{\mathbb{R}  }
\newcommand{\C}{\mathbb{C}}
\newcommand{\I}{\mathbb{I}}
\newcommand{\gn}{\bar{g}_{\textnormal{N}}}
\newcommand{\gi}{\bar{g}_{\textnormal{I}}}
\newcommand{\gf}{\bar{g}_{\textnormal{F}}}
\newcommand{\phitot}{\phi_{\textnormal{tot}}}
\newcommand{\pmax}{P_{\textnormal{max}}}
\newcommand{\pn}{\bar{P}_{\textnormal{N}}}
\newcommand{\dpn}{\overline{\Delta P}_{\textnormal{N}}}
\newcommand{\pno}{\bar{P}_{0,\textnormal{N}}}
\newcommand{\lv}{\left\lvert}
\newcommand{\rv}{\right\rvert}
\newcommand{\llv}{\left\lVert}
\newcommand{\rrv}{\right\rVert}

\newtheorem{definition}{\textbf{Definition}}
\newtheorem{theorem}{\textbf{Theorem}}
\newtheorem{corollary}{\textbf{Corollary}}
\newtheorem{proposition}{\textbf{Proposition}}
\newtheorem{problem}{\textbf{Problem}}
\newtheorem{lemma}{\textbf{Lemma}}
\newtheorem{remark}{\textbf{Remark}}
\newenvironment{proof}{\textit{Proof}:}{\hfill$\square$}

\begin{frontmatter}

\title{Voltage Collapse Stabilization in Star DC Networks\thanksref{footnoteinfo}} 

\thanks[footnoteinfo]{An older version of this paper was presented at the 2019 American Control Conference.}

\author[jhu-case]{Charalampos Avraam}\ead{cavraam1@jhu.edu},    
\author[jhu-ece]{Enrique Mallada}\ead{mallada@jhu.edu}               

\address[jhu-case]{Department
of Civil and Systems Engineering, Johns Hopkins University}  
\address[jhu-ece]{Department
	of Electrical and Computer Engineering, Johns Hopkins University}                                                                 


\begin{abstract}                          

Voltage collapse is a type of blackout-inducing dynamic instability that occurs when power demand exceeds the maximum power that can be transferred through a network. The traditional (preventive) approach to avoid voltage collapse is based on ensuring that the network never reaches its maximum capacity. However, such an approach leads to inefficient use of  network resources and does not account for unforeseen events.
To overcome this limitation, this paper seeks to initiate the study of \textit{voltage collapse stabilization}, {i.e.,} the design of load controllers aimed at stabilizing the point of voltage collapse.
%
We formulate the problem of voltage stability for a star direct current network as a dynamic problem where each load seeks to achieve a constant power consumption by updating its conductance as the voltage changes. We introduce a voltage collapse stabilization controller and show that 
the high-voltage equilibrium is stabilized. More importantly, we are able to achieve proportional load shedding under extreme loading conditions. We further highlight the key features of our  controller using numerical illustrations. 

\end{abstract}

\end{frontmatter}

\section{Introduction}\label{sec:intro}

Loads with constant power consumption tend to dynamically reduce their effective impedance as a means to compensate for a reduction on the power supplied by the network \cite{vvc1998}. However, when the network power transfer capacity is met, further reduction on the effective impedance results in a greater gap between the power supplied to the load and the power demanded by it. This continuous update on the effective impedance  results in an abrupt voltage drop that leads to Voltage Collapse (VC)~\cite{kundur1994,cigre2004}. From a dynamical systems perspective, VC is the manifestation of a saddle node bifurcation where a stable and an unstable equilibrium coalesce, and disappear. As a result, VC is by definition  a dynamic phenomenon that naturally depends on the load dynamics. 
%

Despite its dynamic nature, 
voltage stability studies 
have traditionally focused on static analyses, based on load flow equations \cite{dobson1992},  ensuring enough generation and transmission capacity to avoid VC \cite{iris1997}.
The vast majority of work focuses on the quantification of generation- and network-side voltage stability margins. Demand-side management tools are limited~\cite{Strbac1998}. Operating within the stability margins ensures the system never reaches its limits. Classical system stability metrics include \cite{obadina1988,tvc1991,dobson1994,bompard1996}, and more recently \cite{vournas2016,simpson-porco2016}. However, existing stability margins implicitly enforce a  trade-off between reliability against VC and efficient use of resources~\cite{Strbac2008}.

Fortunately, the rapid development of power electronics and information technology \cite{bayoumi2015} has the potential to provide enough demand-side controllability that could allow us to consider more alternatives for preventing VC. Here, we introduce one such alternative with the study of \textit{voltage collapse stabilization}. Specifically, we aim to investigate how to use (flexible) demand response to reduce consumption and match network capacity, when total demand exceeds it. In this way we prevent inflexible demand from driving the system to VC. To the best of our knowledge, this work is the first effort to design dynamic, demand-side controllers aimed at preventing VC. Such a control scheme needs to overcome several challenges that arise from the dynamic nature of VC.
 
First, the controller needs to stabilize an operating point that is unstable under normal operating conditions.~\footnote{At a saddle node bifurcation a stable and an unstable equilibrium coalesce, which leads to an unstable equilibrium \cite{khalil2002nonlinear}.} 
Further, the stabilization procedure should be robust to the presence of conventional loads that are not willing to reduce their consumption -- referred here as inflexible loads.
%
Thus, in the presence of both flexible and inflexible loads in a system, we ask the following questions:
\begin{itemize}
\item Is voltage collapse stabilization feasible? 
\item {Can a stabilizing controller distribute efficiently the necessary demand reduction among flexible loads?}
\end{itemize}

In this work, we provide an initial answer to these questions for a direct current (DC) star network. 
We consider a resistive star network where each load seeks to consume a constant amount of power by dynamically updating its conductance using a first order voltage droop. %
Despite its simplicity, this model captures the
fundamental drivers of VC\footnote{The model and our results also extend to a fully reactive alternating current star networks.}.

Indeed, when all loads are inflexible (nominal conditions), we show that if the total demand is smaller than the network capacity, the system has a stable and an unstable equilibrium. Further, when demand exceeds capacity,  we analytically show that the voltage $v(t)\rightarrow0$ as $t\rightarrow\infty$, {i.e.,} the voltage collapses.
Once the setup has been analytically validated, we introduce the voltage collapse stabilization controller that stabilizes voltage collapse, even in the presence of inflexible loads, as long as flexible demand is not depleted. Further, we show that stabilization can be achieved with a proportionally fair allocation of load shedding.

The rest of the paper is organized as follows. Section \ref{sec:preliminaries} introduces our DC network model of constant power loads.
Section \ref{sec:analysis-infl-loads} investigates the properties of a system that comprises only inflexible loads and characterizes the region of stable equilibria. Section \ref{sec:eqlbrm-analysis} describes our voltage collapse stabilization controller and characterizes the equilibria of the new system. Section \ref{sec:stab-analysis} studies the stability of each equilibrium point at different operating conditions. 
We illustrate several features of our controller using numerical simulations in Section \ref{sec:results} and conclude in Section \ref{sec:conclusions}.

\section{Problem Setup}\label{sec:preliminaries}
We consider the star DC network model shown in Figure \ref{fig:nbs},
where $E$ denotes the source voltage, $g_l$ the conductance of a transmission line that transfers power to $n$ loads, and $g_i$ denotes the $i$th load conductance, $i\in N:=\{1,\dots,n\}$. We consider two sets of loads, with $F=\{1,...,n_F\}$ being the set of \textit{flexible} loads ($n_{\text{F}}=|F|$), and $I=\{n_F+1,...,n\}$ being the set of \textit{inflexible} loads ($n_I=|I|$). The set of all loads is $N=I\cup F=\{1,\dots, n\}$.
For any subset $S\subseteq N$, we denote the vector of conductances of loads in $S$ by $g_S\in\mathbb R_{\geq0}^{\lvert S \rvert}$ and
\begin{equation}\label{eq:gs}
	\bar{g}_{\text{S}}(g)=\sum_{i\in S} g_i.
\end{equation}
Throughout the paper it is important to keep track of the dependence of several quantities, e.g., voltages, powers, and their derivatives. This explicit dependence $(
\cdot)$ (e.g., $\bar g_\text{S}(g)$) is highlighted when quantities are introduced, but dropped  later on ({e.g.,} $\bar g_\text{S}$) to improve readability of the formulae. 

\begin{figure}
    \begin{center}
        \includegraphics[width=1.0\columnwidth]{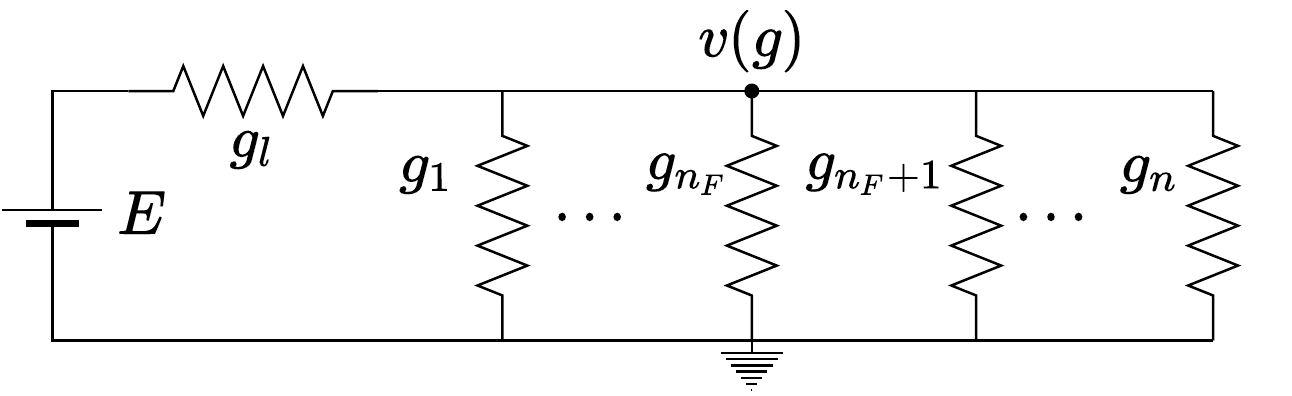}      
    \caption{Star DC network with $n$ dynamic loads.}                    
    \label{fig:nbs}                                                         
    \end{center}                                                            
\end{figure}
Using this notation, we can use Kirchoff's Voltage and Current Laws (KVL and KCL) to compute the voltage applied to each load
\begin{equation}\label{eq:vg}
	v\left(\gn\right)=\frac{E\cdot g_l}{\gn+g_l}.
\end{equation}
where, from \eqref{eq:gs}, $\gn = \sum_{i\in N} g_i$.

Then, the total power consumed by each load $i\in N$ becomes
\begin{equation}\label{eq:Pi}
	P_i(g)=v^2g_i
	=\left(\frac{E\cdot g_l}{\gn+g_l}\right)^2g_i.
\end{equation}
The difference between the power consumed by each load $i\in N$ and its nominal demand $P_{0,i}\geq 0$ is
\begin{equation}\label{eq:DPi}
	\Delta P_i(g) = P_i - P_{0,i}.
\end{equation}
The total power consumed by all loads in the system is
\begin{equation}\label{eq:Pn}
	\pn (\gn)=\sum_{i=1}^nP _i
	= \frac{(E\cdot g_l)^2}{\left(\gn+g_l\right)^2}\gn.
\end{equation}
Similarly, total demand is defined as
\begin{equation}\label{eq:pno}
	\pno = \sum_{i=1}^n P_{0,i},
\end{equation} 
and the difference between total power supply and power demand is 
\begin{equation}\label{eq:dptot}
\dpn \left(\gn\right) = \sum_{i=1}^n \left( P_i - P_{0,i} \right) = 
\pn - \pno,
\end{equation}
where the dependence on $\gn$ follows from \eqref{eq:Pn}.

\noindent
{\textbf{Network Capacity} ($\pmax$):}~Voltage collapse in a DC network can result from 
the network reaching its maximum capacity~\cite{vvc1998}. Therefore, it is of interest to 
compute the maximum value of $\pn(\gn)$ in \eqref{eq:Pn}. A 
straightforward calculation leads to
\begin{equation}\label{eq:dPtgt}
	\frac{\partial}{\partial 
	\gn}\pn\left(\gn\right)=\frac{(Eg_l)^2}{(\gn+g_l)^3}\left( g_l-\gn \right).
\end{equation}

From \eqref{eq:dPtgt}, we can see that $\pn\left(\gn\right)$ is an 
increasing function of $\gn$ whenever $\gn<g_l$, and decreasing when $\gn>g_l$.
Therefore, when $\gn:=g_l$, the maximum power that can be supplied through the line is
\begin{align}\label{eq:Pmax}
	\pmax=\frac{E^2g_l}{4}.
\end{align}

{\bf Dynamic Load Model:}~We assume that each load $i\in N$ has a constant power demand $P_{0,i}$. For an \textit{inflexible load} $i\in I$, this demand $P_{0,i}$ must always be (approximately) satisfied. This is achieved by dynamically changing the conductance $g_i$ in order to change the power supply $P_i(g)$. Following \cite{vournas2008}, we use the following dynamic model
\begin{align}\label{eq:infl}
	\dot{g}_i =&  - ( v^2g_i - P_{0,i} ) = - \Delta P_i, & i\in I.
\end{align}
Notice that $\mathbb{R}^n_{\geq0}$ is invariant, since whenever $g_i=0$ then \eqref{eq:infl} implies that $\dot{g}_i>0$.

For the \textit{flexible loads}, we assume that they are willing to consume less than $P_{0,i}$ whenever $\pno:=\sum_{i\in N}P_{0,i}>P_{\text{max}}$. Thus, our goal is to design a control law
\begin{equation}\label{eq:fl}
	\dot{g}_i = u_i, \quad 	 i\in F,
\end{equation}
where the input $u_i$ is such that in equilibrium
$\Delta P_i(g)=0$ whenever $\pno<P_{\text{max}}$.

{\bf Power Flow Solutions:} Given an equilibrium $g^*$ of \eqref{eq:infl}-\eqref{eq:fl}, there exists a unique voltage $v^*$ and power consumption $P(g^*)=(P_i(g^*)$, $i\in N)$. The pair $(v,P)$ is referred as \textit{power flow solution}. Thus, given the one-to-one relationship between $g$ and the pair $(v,P)$, we refer to $g^*$ as a power flow solution. 
\section{System Analysis with Inflexible Loads}\label{sec:analysis-infl-loads}

Throughout this section we validate our inflexible load model by showing that, when $N=I$ (all loads are inflexible), if the demand $\pno$ exceeds the maximum deliverable power $P_{\text{max}}$ ($\pno>P_{\text{max}}$), the system undergoes VC.
We further characterize the region of stable equilibria of \eqref{eq:infl} when $\pno<P_{\text{max}}$ and motivate the need for coordination to prevent voltage collapse.

\subsection{Voltage Collapse}
 We first show that \eqref{eq:infl} undergoes VC in the overload regime ($\pno>P_{\text{max}}$). To this aim, we provide a formal definition to VC.

 \begin{definition}[Voltage Collapse]
	The system \eqref{eq:infl} undergoes voltage collapse whenever $v(g(t))\rightarrow 0$  as $t\rightarrow+\infty$.
\end{definition}

We can show that voltage collapse occurs from the response of loads to overloading.
\begin{theorem}[Voltage Collapse]\label{thm:vc}
	The dynamic load model \eqref{eq:infl} with $I=N$ undergoes voltage collapse whenever $\varepsilon:=\pno-P_{\text{max}}>0$.
\end{theorem}
\begin{proof}
	We have already pointed out that $\mathbb{R}^n_{\geq0}$ is invariant. Also, it is easy to check that \eqref{eq:infl} is globally Lipschitz on $\mathbb{R}^n_{\geq0}$ since $g_l>0$. Thus, by \cite[Theorem 3.2]{khalil2002nonlinear} there is a unique solution to \eqref{eq:infl}, $g(t)$, that is defined $\forall t\geq0$. Now, consider the function $V(g) = \sum_{i\in N}g_i$, and let $S^+_V(a)=\{g\in\mathbb{R}^n_{\geq0}: V(g)\leq a\}$.
	By taking the time derivative of $V$ we get
	\begin{align*}
		\dot V(g) &= \sum_{i=1}^n \dot g_i = -\sum_{i=1}^n P_i(g)-P_{0,i} \\
		&\geq \pno-P_{\text{max}} =\varepsilon>0.
	\end{align*}
	Therefore, $\forall a\geq 0, $ if $g(0)\in S_V^+(a)$, the solution $g(t)$ escapes $S_V^+(a)$ in finite time. It follows then that $\llv g(t)\rrv\rightarrow \infty$ as $t\rightarrow \infty$, i.e.,  $\gn(t)$ grows unbounded. Therefore, by \eqref{eq:vg} $v(\gn(t))\rightarrow 0$ and the system's voltage collapses.
\end{proof}
\subsection{Equilibrium Analysis with Inflexible Loads}
Since \eqref{eq:infl} undergoes VC when $\pno>\pmax$, in this section we will study the properties of the equilibria of \eqref{eq:infl} when $\pno<\pmax$. The following lemma will allow us to characterize these equilibria.
\begin{lemma}[Intermediate Value Theorem \cite{rudin1976}]\label{lem:ivt}
	Let $f:\R\rightarrow [a,b]\subset\R$, continuous function. For any $\psi\in\left(f(a),f(b)\right)$ there exists $\xi\in[a,b]$ such that $f(\xi)=\psi$.
\end{lemma} 
%
\begin{theorem}[Equilibrium Characterization of \eqref{eq:infl}]\label{thm:char-infl}
    Let $I=N$. Then, the equilibria of system \eqref{eq:infl-cond} have the following properties
    
    \begin{enumerate}[label=\alph*)]
        \item When $0<\pno<\pmax$, the system \eqref{eq:infl-cond} has two equilibria $g_1^*,g_2^*\in\R^n$ such that
    	\begin{equation}\label{eq:infl-cond}
    		\gn(g_1^*)< g_l<\gn(g_2^*).
    	\end{equation}
        \item When $\pno=0$, the system \eqref{eq:infl-cond} has one equilibrium $g_1^*=0$.
        \item When $\pno=\pmax$, the system \eqref{eq:infl-cond} has one equilibrium $g^*$ such that $\gn (g^*)=g_l.$
    \end{enumerate}
	 	
\end{theorem} 
\begin{proof}
	Let $g^*$ an equilibrium of \eqref{eq:infl}. Summing \eqref{eq:infl} for all $i\in N$ together with \eqref{eq:dptot} and \eqref{eq:Pn} gives
	
	\begin{equation}\label{eq:Dpnstar}
	    \dpn(\gn^*)
		=\frac{\left(Eg_l\right)^2\gn^*}{(\gn^*+g_l)^2}-\pno=0.
	\end{equation}
		
	\begin{enumerate}[wide, labelindent=0pt, label=\alph*)]
	\item Since $\pno>0$, we multiply both sides of \eqref{eq:Dpnstar} by $-\frac{(\gn^*+g_l)^2}{\pno}$ to get
	\begin{equation}\label{eq:poly-gn}
		\left(\gn^*\right)^2 + \left(2g_l-\frac{\left(Eg_l\right)^2}{\pno}\right)\gn^*+g_l^2=0.
	\end{equation}
	This is a second order polynomial in $\gn^*$ and has at most two real roots $\gn^{(1)}, \gn^{(2)}$.
	
    Equation \eqref{eq:poly-gn} has 2 distinct real roots if and only if 
	    \begin{align}
    		&
    		\left(2g_l-\frac{\left(Eg_l\right)^2}{\pno}\right)^2 - 4g_l^2 >0\;\iff\nonumber\\
     		& \pno<\frac{\left(Eg_l\right)^2}{4g_l}=\frac{E^2 g_l}{4}\overset{\eqref{eq:Pmax}}{=}\pmax. 
    		\label{eq:poly-gn-condition}
	    \end{align}
	     Thus, when $\pno<\pmax$ then \eqref{eq:poly-gn} has two equilibria  $g_1^*$ and $g_2^*$  whose sum, $\gn^{(1)}$ and $\gn^{(2)}$, are distinct. We assume, without loss of generality,  $\gn^{(1)}<\gn^{(2)}$. Consider the function $\dpn(\gn)$ given in \eqref{eq:Dpnstar}. It follows from \eqref{eq:dptot} that
	     \begin{equation*}
		 \begin{split}
			\frac{\partial \dpn(\gn)}{\partial \gn}=\frac{\partial 
			\pn(\gn)}{\partial \gn}   \overset{\eqref{eq:dPtgt}}{=}
			\begin{cases}
			>0,  \quad \gn<g_l    \\
			<0,  \quad \gn>g_l
			\end{cases}
		 \end{split}
	    \end{equation*}
    	We can prove by contradiction that the two roots of $\dpn(\gn)$ satisfy \eqref{eq:infl-cond}. If \eqref{eq:infl-cond}  is not met, then both equilibria $\gn^{(1)},\gn^{(2)}$ are either in  $(0,g_l)$ or in $(g_l,+\infty)$. In the first interval $\dpn$ is strictly increasing, whereas in the second $\dpn$ is strictly decreasing. In this case,  
    	$$0=\Delta\pn\left(\gn^{(1)}\right) \neq \Delta\pn\left(\gn^{(2)}\right)=0, $$ contradiction. Therefore, \eqref{eq:infl-cond} holds. 
    	
    	Finally, we show that \eqref{eq:infl} has exactly two equilibria by computing all possible equilibria. Substituting $\gn^{(i)}$, $i=1,2$, into \eqref{eq:vg} results in
    	\begin{align*}
    		v\left(\gn^{(i)}\right)=\frac{Eg_l}{\gn^{(i)}+g_l}\overset{\eqref{eq:infl}}{\Rightarrow} g_{j}^{(i)}=\frac{P_{0,i}}{v^2\left(\gn^{(i)}\right)}, \, i=1,2,
    	\end{align*}
    	for all $j\in N$.
    	
    	\item Since $P_{0,i}\geq 0$ $\forall i\in N$, then the trivial case where $\pno=0$ arises only when $P_{0,i}=0$ $\forall i\in N$. In that case \eqref{eq:infl} immediately leads to $g_i^*=0$ $\forall i\in N$. We can also check that when $g^*_2\geq 0$ and $\gn(g_2^*)=+\infty$, then $g_2^*$ also satisfies \eqref{eq:Dpnstar}. However, $g_2^*$ is not finite.
    	
    	\item 	In the non-trivial case where $\pno=\pmax\overset{\eqref{eq:Pmax}}{=}\frac{E^2g_l}{4}$, the discriminant of \eqref{eq:poly-gn} becomes
	\begin{equation*}
		\Delta=\!\left(\!2g_l-\frac{\left(Eg_l\right)^2}{\pmax}\!\right)-4g_l^2\!\overset{\eqref{eq:Pmax}}{=}\! \left(\!2g_l-\frac{\left(Eg_l\right)^2}{\frac{E^2g_l}{4}}\!\right)-4g_l^2=0.
	\end{equation*}
	
	Then, \eqref{eq:poly-gn} becomes
	$$
		\left(\gn^*\right)^2 + \left(2g_l-\frac{\left(Eg_l\right)^2}{\frac{E^2g_l}{4}}\right)\gn^*+g_l^2=\left(\gn^*-g_l\right)^2=0,
	$$
	which has the unique solution $\gn^*=g_l$. Substituting $\gn^*=g_l$ into \eqref{eq:vg}  leads to $v^*=\frac{E}{2}$. Plugging $v^*$ into \eqref{eq:infl} and solving for $g_i^*$ gives the unique solution $g_i^*=\frac{4 P_{0,i}}{E^2}$ $\forall i\in N$.
	\end{enumerate}\end{proof}
\subsection{Stable Region Characterization}

We end this section by characterizing the region of the state space that admits locally asymptotically stable equilibrium points of \eqref{eq:infl}. 
The following lemma will be instrumental for this task.
\begin{lemma}[Rank-1 plus scaled identity matrix]\label{lem:evals-r1+identity}
	For  $w\in\C^{n\times 1}$, $\mathbf{1}_n\in\R^n$ a column vector of all ones, $\I_n\in\C^{n\times n}$ the $n\times n$ identity matrix, and $q\in\C$, the eigenvalue-eigenvector pairs $\left(\rho,x\right)\in\C\times\C^n$ of the Rank-1 Plus Scaled Identity (RPSI) matrix $B=w\mathbf{1}_n^T+q\I_n$ are
	\begin{equation}\label{eq:evecs-evals-r1+identity}
		\left(\rho,x\right)=\begin{cases}
			\left(q,e_1-e_i\right), & i\in\left\{1,..,n-1\right\};\\
			\left(\sum_{i=1}^n w_i+q,w\right),& i=n. \\
		\end{cases}
	\end{equation}
\end{lemma}

We are now ready to characterize the stable region of \eqref{eq:infl}.
For that, we consider the set
\begin{align}\label{eq:reg-M}
M :=\left\{g\in\R_{\geq 0}^n: \sum_{i\in N} g_i < g_l\right\}.
\end{align}

\begin{theorem}[Stable Region Characterization]\label{thm:char-stab-reg}
	A hyperbolic equilibrium\footnote{An equilibrium is hyperbolic if the Jacobian of the system at the equilibrium point is nonsingular.} $g^*$ of \eqref{eq:infl} is stable if and only if $g^*\in M$.
\end{theorem}
\begin{proof}
	Let $g^*$ be an equilibrium of \eqref{eq:infl}, {i.e.,} $\Delta P_i(g^*)=0$ for all $i \in N$, $v^*:=v(g^*)$ and $\gn^*:=\gn(g^*)$. The Jacobian of the system \eqref{eq:infl} evaluated at $g^*$ is given by
	\begin{equation}\label{eq:J}
		J\left(g^*\right) =  \frac{2{v^*}^2}{\gn^*+g_l}g^*\mathbf{1}_n^T-{v^*}^2 \I_n.
	\end{equation}
	
	Moreover, $J\left(g^*\right)$ is a RPSI matrix. Therefore, we can compute its eigenvalues by substituting $q=-{v^*}^2$ and $w=\frac{2{v^*}^2}{\gn^*+g_l}g^*$ in Lemma \ref{lem:evals-r1+identity}: 
	\begin{align*}
	\lambda_i(J)=
	\begin{cases}
		-{v^*}^2, &  i= 1,..,n-1; \\
		\frac{2{v^*}^2}{\gn^*+g_l}\gn^*-{v^*}^2,&i=n.
	\end{cases} 
	\end{align*}
	We can now prove the statement of the theorem.
	
	($\Rightarrow$) If $g^*$ is an asymptotically stable hyperbolic equilibrium, then $J(g^*)$ is Hurwitz and thus:
	$\lambda_n(J)<0 \Rightarrow {v^*}^2\left(\frac{2\gn^*}{\gn^*+g_l}-1\right)<0 \Rightarrow \gn^*<g_l$. \\
	($\Leftarrow$) If $g^* \in M$, then: $\lambda_n(J)<0$. Since all eigenvalues of $J(g^*)$ are negative, by Lyapunov's Indirect Method~\cite[Theorem 3.5]{khalil2002nonlinear} $g^*$ is asymptotically stable. \end{proof}

\begin{figure}
    \begin{center}
        \includegraphics[width=1.0\columnwidth]{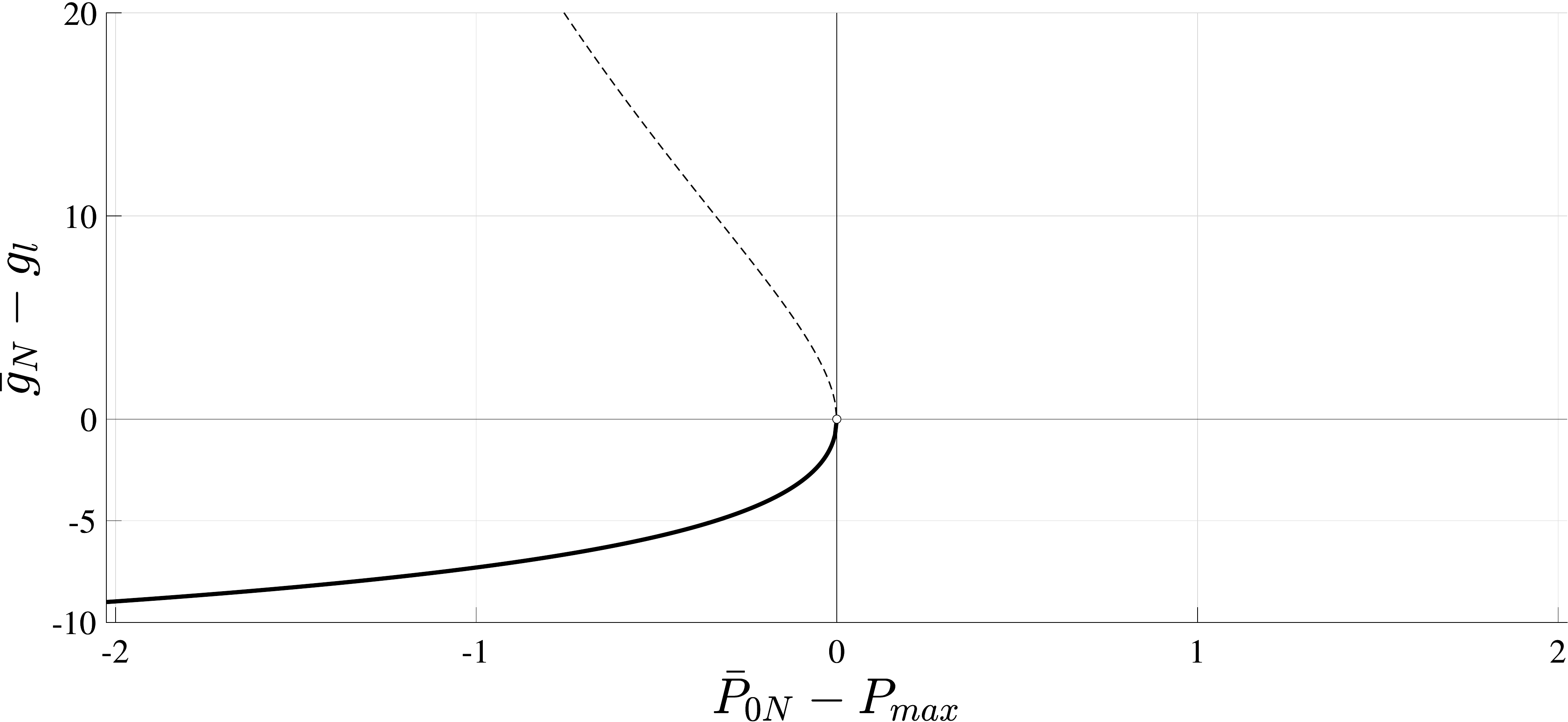}      
    \caption{Bifurcation Diagram of \eqref{eq:infl}. As total demand approaches the capacity of the line ($\pno=\pmax$), the two equilibria of \eqref{eq:infl} converge to $\gn(g_1^*)=\gn(g_2^*)=g_l$, coallesce at $\pno=\pmax$, and disappear for $\pno>\pmax$. Therefore, the system \eqref{eq:infl} undergoes a saddle-node bifurcation at $\pno=\pmax$. }                    
    \label{fig:saddlenode}                                                         
    \end{center}                                                            
\end{figure}

The above analysis shows how the load model \eqref{eq:infl} captures the underlying principles of voltage collapse. More precisely, whenever $\pno<\pmax$, the system has two equilibria, one of which allows each load to stably meet its own demand (provided that $\gn^*<g_l$) by updating its own conductance (Theorem \ref{thm:char-stab-reg}). Figure \ref{fig:saddlenode} shows that the system undergoes a saddle-node bifurcation at $\pno=\pmax$ and the equilibria disappear for $\pno>\pmax$. At any operating point $\pno>\pmax$, the network is unable to provide sufficient power to all loads. As a result, every load seeks to increase its conductance, leading to voltage collapse (Theorem \ref{thm:vc}). Notably, it is this individual (selfish) action of each loads that leads to the overall system collapse, a behavior akin to the game-theoretic notion of the Tragedy of the Commons \cite{faysse2005}. This observation motivates investigating coordination strategies that address the selfish behavior. In this way, we seek to transform the type of bifurcation happening at $\pno=\pmax$ such that a locally asymptotically stable equilibrium exists, at the bifurcation point, and for operating points $\pno>\pmax$. 
\section{Voltage Collapse Stabilization}\label{sec:eqlbrm-analysis}
In this section we illustrate that it is possible to prevent voltage collapse by allowing a certain level of coordination among a subset of \textit{flexible} loads that are willing to coordinate their actions.

\subsection{VCS Controller}
We now consider the system \eqref{eq:infl}-\eqref{eq:fl}, with $F\not=\emptyset$.
For each flexible load $i\in F$, we seek to control its consumption using 
\begin{subequations}\label{eq:scc-dynamics-0}
	\begin{alignat}{2}
		\label{eq:scc-dg-f-0} u_i =& -\Delta P_i(g) - \kappa_i \cdot\phi + b\left(\hat{g}_i - g_i\right)
		, & \forall i\in F \\
		\label{eq:scc-dghat-0} \dot{\hat{g}}_i =& -a\left(\hat{g}_i-g_i\right), & \forall i\in F  \\
		\intertext{where the shared state $\phi$ evolves according to}
		\label{eq:scc-dphi-0}\dot{\phi}=&  \dpn(\gn) \cdot  
		\frac{\partial}{\partial \gn} \dpn\left(\gn\right).& 
	\end{alignat}
\end{subequations}
A few explanations are in order. The parameters $\kappa_i>0$, $i\in F$, in \eqref{eq:scc-dg-f-0} are called load shedding parameters, and will be used to control the fraction of  excess consumption that is shed by load $i\in F$.
We will further define $\bar{\kappa}:=\sum_{i\in F} \kappa_i$, and $a,b>0$. We will see in the following section that $a$ and $b$ affect the stability of the equilibria of the full system (c.f. \eqref{eq:scc-dynamics}). The second term in \eqref{eq:scc-dg-f-0} ensures proportional load shedding. Finally, the third term in \eqref{eq:scc-dg-f-0} together with \eqref{eq:scc-dghat-0} introduces dynamic damping, and is akin to regularization techniques present in saddle-point dynamics~\cite{you2020}. 
If the demand of an inflexible load is zero then we omit it. Therefore, the above dynamics implicitly assume that  $P_{0,j}>0$ for all $j\in I$.

For the vector $(\begin{array}{ccc} g, & \phi, & \hat{g} \end{array})\in\R^n_{\geq 0}\times \R \times\R^{n_F}$ of state variables, we obtain the following closed-loop dynamics:
\begin{subequations}\label{eq:scc-dynamics}
	\begin{alignat}{2}
		\label{eq:scc-dg-f} \dot{g}_i =& -\Delta P_i(g) - \kappa_i \cdot\phi + b\left(\hat{g}_i - g_i\right)
		, & \forall i\in F \\
		\label{eq:scc-dg-i}\dot{g}_i =& -\Delta P_i(g), & \forall i\in I \\
		\label{eq:scc-dphi}\dot{\phi}=&  \dpn(\gn) \cdot  
		\frac{\partial}{\partial \gn} \dpn\left(\gn\right),& 
		  \\
		\label{eq:scc-dghat} \dot{\hat{g}}_i =& -a\left(\hat{g}_i-g_i\right), & \forall i\in F
	\end{alignat}
\end{subequations}
\subsection{Equilibrium Characterization}
We end this section by providing a characterization of the equilibria of \eqref{eq:scc-dynamics}.

\begin{theorem}[Equilibrium Characterization of \eqref{eq:scc-dynamics}]\label{thm:char-scc}
The system \eqref{eq:scc-dynamics} has two sets of equilibria, the \emph{load satisfaction} set 
\begin{equation}\label{eq:fcc-eqset-1}
    \begin{split}
    	\mathcal{E}_s = &\Bigg\{  \left(\begin{array}{ccc} g^*,&\phi^*,&\hat{g}^* \end{array}\right)\in \R_{\geq 0}^n\times\R\times\R^{n_F} : \Delta P_i(g^*)=0 \\
    	&\quad \forall i\in N, \,\, \phi^*=0, \quad \hat{g}^*=g^*  \Bigg\},
	\end{split}	
\end{equation}
where the demand of all loads is satisfied, and the \emph{proportional allocation} set 
\begin{equation}\label{eq:fcc-eqset-2}
	\begin{split}
		\mathcal{E}_p = \Big\{&\left(\begin{array}{ccc} g^*,&\phi^*,&\hat{g}^* \end{array}\right)\in \R_{\geq 0}^n\times\R \times \R^{n_F} : \\
		&\Delta P_i(g^*)=\frac{\kappa_i}{\bar{\kappa}} \left(\pmax-\pno\right) \, \forall i\in F,\\
		&\Delta P_i(g^*)=0 \,\, \forall i\in I, \,\phi^*  =\frac{\pno - \pmax}{\bar{\kappa}}, \\ &\hat{g}^*=g^*  \Big\},
	\end{split}
\end{equation}
where the difference $ \pno-\pmax $ is proportionally allocated among flexible loads.
\end{theorem}
\begin{proof}
If $(g^*,\phi^*,\hat{g}^*)$ is an equilibrium of \eqref{eq:scc-dynamics}, then \eqref{eq:scc-dghat} implies that $\hat{g}^*_i=g_i^*$ for all $i\in F$. Moreover, \eqref{eq:scc-dg-f}, \eqref{eq:scc-dg-i} lead to
\begin{equation}\label{eq:scc-cond-g}
	\Delta P_i(g^*)=-\kappa_i\phi^* \,\, \forall i\in F,\quad \Delta P_i(g^*) = 0  \,\, \forall i\in I.
\end{equation}
Last, by \eqref{eq:scc-dphi} either $\dpn(\gn^*)=0$ or $\frac{\partial \dpn(\gn)}{\partial \gn}\Big\rvert_{\gn^*}=0$.
\begin{itemize}[leftmargin=*]
	\item If $\dpn(\gn^*)=0$, then by \eqref{eq:scc-cond-g}
	$$
	    0=\dpn(\gn^*)=\sum_{i\in N} \Delta P_i(g^*)\overset{\eqref{eq:scc-cond-g}}{=} -\bar{\kappa} \phi^* \Rightarrow 
		\phi^* = 0.
	$$
	Substituting $\phi^*=0$ back into \eqref{eq:scc-cond-g} gives $ \Delta P_i\left(g^*\right)=0$ for all $i\in N$. Hence, the first set of equilibria is
	\begin{equation}\
    	\begin{split}
    		\mathcal{E}_s = &\Bigg\{  \left(\begin{array}{ccc} g^*,&\phi^*,&\hat{g}^* \end{array}\right)\in \R_{\geq 0}^n\times\R\times\R^{n_F} : \Delta P_i(g^*)=0 \\
    		&\quad \forall i\in N, \,\, \phi^*=0, \quad \hat{g}^*=g^*  \Bigg\}.
    	\end{split}	
    \end{equation}
	When $\pno<\pmax$, by Theorem \ref{thm:char-infl} there exist two equilibrium vectors of conductances $g^*$ such that $\Delta P_{i}(g^*)=0$ $\forall i\in N$. Therefore, the load satisfaction set $\mathcal{E}_s$ comprises two equilibria $\left(\begin{array}{ccc}g^*_1, &\phi^*, & \hat{g}_1^*\end{array}\right)$, $\left(\begin{array}{ccc}g^*_2, &\phi^*, & \hat{g}_2^*\end{array}\right)$ such that $\gn\left(g^*_1\right)<g_l<\gn\left(g^*_2\right)$. 
	\item If $\frac{\partial \dpn(\gn)}{\partial 
	\gn}\Big\rvert_{\gn^*}=0$, then by \eqref{eq:Pmax} 
	$\gn^*=g_l$ and $\pn^*=\pmax$, {i.e.,}
	\begin{align*}
		&\sum_{i\in N} 
		\left(P_i\left(g^*\right) - P_{0,i} \right)=\pmax 
		-\sum_{i\in N} P_{0,i} \\
		\overset{\eqref{eq:DPi}}{\Rightarrow}& \sum_{i\in F} \Delta P_i\left(g^*\right)+\sum_{i\in I} \Delta P_i\left(g^*\right)=\pmax-\pno \\
	    \overset{\eqref{eq:scc-cond-g}}{\Rightarrow}& -\sum_{i\in F}\kappa_i\phi^* = \pmax-\pno \Rightarrow \phi^*=\frac{ \pno-\pmax}{\bar{\kappa}}.
	\end{align*}
	Substituting $\phi^*=\frac{\pno-\pmax}{\bar{\kappa}}$ back into \eqref{eq:scc-cond-g} gives $ \Delta P_i(g^*)=\frac{\kappa_i}{\bar{\kappa}}\left(\pmax-\pno\right)$ $\forall i\in F$. Hence, the second set of equilibria is
	\begin{equation}
		\begin{split}
			\mathcal{E}_p = \Big\{&\left(\begin{array}{ccc} g^*,&\phi^*,&\hat{g}^* \end{array}\right)\in \R_{\geq 0}^n\times\R \times \R^{n_F} : \\
			&\Delta P_i(g^*)=\frac{\kappa_i}{\bar{\kappa}} \left(\pmax-\pno\right) \, \forall i\in F,\\
			&\Delta P_i(g^*)=0 \,\, \forall i\in I, \,\phi^*  =\frac{\pno - \pmax}{\bar{\kappa}}, \\ &\hat{g}^*=g^*  \Big\}.
		\end{split}
	\end{equation}
\end{itemize}
   Since $\gn^*=g_l$ when $\left(g^*,\,\, \phi^*, \,\, \hat{g}^*\right)\in\mathcal{E}_p$, then \eqref{eq:vg} defines a unique voltage $v^*=\frac{E}{2}$, which in turn leads to the unique vector of conductances  
	   \begin{equation}
	    \begin{split}
	    	&\Delta P_i(g^*)=\begin{cases}
	    	  	\frac{\kappa_i}{\bar{\kappa}} \left(\pmax-\pno\right), &\quad \forall i\in F \\
	    	  	0, &\quad \forall i\in I
	    	  	\end{cases} \\
	    	\Rightarrow &g_i^*=\begin{cases}
	    	    \frac{1}{\left(\frac{E}{2}\right)^2}\left(P_{0,i}+\frac{\kappa_i}{\bar{\kappa}} \left(\pmax-\pno\right)\right), &\quad \forall i\in F \\
	    	    \frac{1}{\left(\frac{E}{2}\right)^2}P_{0,i}, &\quad \forall i\in I
	    	    \end{cases}
	    \end{split}
	   \end{equation}
    Therefore, for any demand $\pno\geq 0$, the proportional allocation set $\mathcal{E}_p$ is a singleton $ \mathcal{E}_p=\left\{\left(g^*,\,\,\phi^*,\,\,\hat{g}^*\right)\right\}$. 

\end{proof}

An equilibrium of the \textit{load satisfaction} set $g^*\in\mathcal{E}_s$ ensures that all loads $i\in N$ satisfy their individual demand $P_i(g^*)=P_{0,i}$. Similarly, the equilibrium of the \textit{proportional allocation} set $g^*\in\mathcal{E}_p$ ensures that the amount $\pno-\pmax$ is distributed among flexible loads proportional to the load-shedding parameter $\kappa_i$, $i\in F$.

\section{Stability Analysis of VCS Controller}\label{sec:stab-analysis}
In this section we evaluate the stability of the equilibria of \eqref{eq:scc-dynamics} under different loading conditions. To assess stability of each equilibrium we study the linearised model at that point. 

\subsection{Stability analysis of efficient operating points}\label{subsec:vcs-stability-of-efficient-points}
The Jacobian of \eqref{eq:scc-dynamics} is
\begin{equation}\label{eq:jacob-scc}
	\begin{split}
		&J_{SC}\left(g\right)= \\
		&\left[\begin{array}{ccc} J(g) - \left[\begin{array}{c|c} b\I_F & \mathbf{0}_{n_F} \mathbf{0}_{n_I}^T \\ \hline \mathbf{0}_{n_I}\mathbf{0}_{n_F}^T & \mathbf{0}_{n_I} \mathbf{0}_{n_I}^T
		\end{array}\right] & \left[\begin{array}{c} -\kappa \\ \hline \mathbf{0}_{n_I}
		\end{array} \right]& \left[\begin{array}{c} b\I_{n_F} \\ \hline \mathbf{0}_{n_I}
		\end{array}\right] \\ c \mathbf{1}_n^T & 0 & \mathbf{0}_{n_I}^T \\ \left[\begin{array}{c|c} a\I_{n_F} & \mathbf{0}_{n_F}\mathbf{0}_{n_I}^T
		\end{array}\right] & \mathbf{0}_{n_F} & -a\I_{n_F} \end{array}\right]
	\end{split}
\end{equation}
with $\mathbf{0}_n\in\R^{n}$ column vector of all zeros.
\begin{lemma}[Eigenvalue Characterization of \eqref{eq:jacob-scc}]\label{lem:jsc-evals} Let $0<m<M<+\infty$, and consider $b>0$ in \eqref{eq:scc-dg-f} and $a>0$ in \eqref{eq:scc-dghat} such that
\begin{equation}\label{eq:stab-cond}
   a\in (0,v_{\min}^2), \quad v_{\min}:= \min\limits_{g:m\leq g_i \leq M, \,\, \forall i\in I} \quad v
\end{equation}
Matrix $J_{SC}\left(g\right)$ as in \eqref{eq:jacob-scc} has $2n_F-2$ eigenvalues that are roots of
	\begin{align}\label{eq:jacob-scc-cond1}
		\lambda^2 +\left(a+b+v^2\right)\lambda + av^2=0,
	\end{align}
	$n_I-1$ eigenvalues that are $-v^2$
	and four eigenvalues that are roots of
	\begin{subequations}\label{eq:jacob-scc-cond2}
		\begin{align}
			\label{eq:jacob-scc-cond2-main}&\lambda^4+\alpha_3\left(g\right)\lambda^3+\alpha_2\left(g\right)\lambda^2+\alpha_1\left(g\right)\lambda+\alpha_0\left(g\right)=0,\quad \,\,\,\, \\
			\label{eq:jacob-scc-cond2-alpha-0}& \alpha_0\left(g\right)=\bar{\kappa} ac v^2 \\
			\label{eq:jacob-scc-cond2-alpha-1}& \alpha_1\left(g\right)= a v^2 \tilde{\lambda}+\bar{\kappa} c\left(a+v^2\right)  \\
			\label{eq:jacob-scc-cond2-alpha-2}& \alpha_2\left(g\right)= \bar{\kappa} c +v^2\left(a+b \right)-\frac{2bv^2\gi}{\gn+g_l} 
			+ \tilde{\lambda}\left(a+v^2\right)  \\
			\label{eq:jacob-scc-cond2-alpha-3}& \alpha_3\left(g\right)=  b+a+v^2+\tilde{\lambda} 
		\end{align}
	\end{subequations}
	where
	\begin{equation}\label{eq:tilde-lambda} 
    \tilde{\lambda} \left(\gn\right) = v^2 - \frac{2v^2\gn}{\gn+g_l}
\end{equation}
\begin{equation}\label{eq:c(gt)}
	\begin{split}
		&c(\gn)= \left(\frac{ v^2\left(g_l-\gn\right) 
		}{\gn+g_l}\!\right)^2+ 
		\dpn\frac{2v^2(\gn-2g_l)}{(\gn+g_l)^2}
	\end{split}
\end{equation}

\end{lemma}
\begin{proof} We provide a complete proof in Appendix \ref{appdx:proof-of-lemmas}.\end{proof}
%
%
%
\begin{theorem}[Stability of VCS Controller]\label{thm:stab} 
Consider the system \eqref{eq:scc-dynamics} with $a>0$ satisfying \eqref{eq:stab-cond}. For $a$ in \eqref{eq:scc-dghat} sufficiently small and $b>0$ in \eqref{eq:scc-dg-f}, the following  hold:
\begin{enumerate}[leftmargin=*]
	\item  When $ \pno > P_{\text{max}}$, the unique equilibrium $\left(g^*,\,\,\phi^*,\,\,\hat{g}^*\right)\in\mathcal{E}_p$ is locally asymptotically stable.
	\item  When $\pno < P_{\text{max}}$, the equilibrium $\left(g^*,\,\,\phi^*,\,\,\hat{g}^*\right)\in \mathcal{E}_s\cap \textnormal{cl}(M)=\left\{\left(g^*,\,\,\phi^*,\,\,\hat{g}^*\right)\right\}$ is locally asymptotically stable.
\end{enumerate}
\end{theorem}
\begin{proof} To assess stability of \eqref{eq:scc-dynamics} we will study the 
linearized system. The Jacobian of the system is \eqref{eq:jacob-scc}. By Theorem 
\ref{lem:jsc-evals}, $J_{SC}(g)$ has $n_I-1$ eigenvalues that are $-v^2<0$ and 
$2n_F-2$ eigenvalues that are roots of 
\eqref{eq:jacob-scc-cond1}. When $v^2>0$, by the Routh-Hurwitz criterion, \eqref{eq:jacob-scc-cond1} has two roots $\lambda_1,\lambda_2$ with negative real parts. Therefore, the stability of the system depends on the last four 
eigenvalues that are roots of \eqref{eq:jacob-scc-cond2}. The Routh-Hurwitz table of \eqref{eq:jacob-scc-cond2} is 
\begin{subequations}\label{eq:r-h}
	\begin{align}
		\label{eq:r-h-table}&\begin{array}{r|cccc}	
			s^4	&	1 										& \alpha_2(g) 				&	\alpha_0(g)		&	0 \\
			s^3	&	\alpha_3(g)								& \alpha_1(g) 				& 0					& 	0 \\
			s^2	&	b_1(g)									& b_2(g) 					& 0					& 	0 \\
			s^1	&	c_1(g)									& 0  						& 0					& 	0 \\
			s^0 &	d_1(g)									& 0 						& 0					& 	0
		\end{array} \\
		\label{eq:r-h-b1}&b_1(g)=\frac{\alpha_3(g)\alpha_2(g)-1\cdot\alpha_1(g)}{\alpha_3(g)} 	\\
		\label{eq:r-h-b2}&b_2(g)=\frac{\alpha_3(g)\alpha_0(g)-1\cdot 0}{\alpha_3(g)}=\alpha_0(g) 	\\
		\label{eq:r-h-c1}&c_1(g)=\frac{b_1(g)\alpha_1(g)-\alpha_3(g)b_2\left(g\right)}{b_1(g)} 	\\
		\label{eq:r-h-d1}&d_1(g)=\frac{c_1(g)b_2(g)}{c_1(g)}=b_2(g)=\alpha_0(g) 	
	\end{align}
\end{subequations}
We can immediately check that
\begin{align*}
	\alpha_3(g)=& -\frac{2v^2 \gn}{\gn+ g_l}+ 2v^2 +a
	+b > 0; \\
	d_1(g)
	\overset{\eqref{eq:c(gt)}}{=}&\bar{\kappa} a v^2 \Bigg[\left(\frac{ v^2 }{\gn+g_l}\left(g_l-\gn\right)
	\right)^2+ \\
	&\qquad +\dpn \frac{2v^2}{(\gn+g_l)^2}\left(\gn- 
	2g_l\right)\Bigg]
\end{align*}
Moreover, we can rewrite \eqref{eq:r-h-b1} and \eqref{eq:r-h-c1} respectively as
\begin{equation}\label{eq:r-h-b1-full}
	\begin{split}
		b_1( g)=& \bar{\kappa}c\frac{b+\tilde{\lambda}}{b
		+a+v^{2}+\tilde{\lambda}}+v^{2}\left(a+	b\right)-\\
		&-\frac{2bv^{2}\gi}{\gn+g_l}+\tilde{\lambda}\left(v^{2}+
		\frac{a
		\left(b+a+\tilde{\lambda}\right)}{b
		+a
		+v^{2}+\tilde{\lambda}} \right)	
	\end{split}
\end{equation}
\begin{equation}\label{eq:r-h-c1-full}
	\begin{split}
		c_1( g)=&\bar{\kappa}c
		\frac{\left(a+v^{2}\right)  b_1\left( g \right)-\left(b +a+v^{2}
		\right)
		av^{2}}{b_1\left(g\right)} + \\
		&+\tilde{\lambda} a v^{2}\left(1-
		\frac{\bar{\kappa} c^{2}}{b_1\left(g\right)}
		\right)
	\end{split}
\end{equation}

\begin{enumerate}[leftmargin=*]
	\item When $\pno>\pmax$, then the \textit{load satisfaction} set is empty, i.e., $\mathcal{E}_s=\emptyset$, and the \textit{proportional allocation} set is a singleton, $\mathcal{E}_p=\left\{\left(g^*,\,\,\phi^*,\,\,\hat{g}^*\right)\right\}$. Moreover, 
	$$
	    \eqref{eq:Pmax}\Rightarrow \gn^*=g_l, \quad \eqref{eq:tilde-lambda}\Rightarrow  \tilde{\lambda}^*=\tilde{\lambda}\left(g_l
	\right)=0,   
	$$
    $$
	    \eqref{eq:c(gt)} \Rightarrow	c^*=\left(\pno-\pmax \right)
		\frac{v^{*^2}}{2 g_l} > 0 \overset{\eqref{eq:r-h-d1}}{\Rightarrow} d_1 \left(g^*
		\right)>0.
	$$
	Notice that 
	$$
	\gi^* \leq \gn^*=g_l \Rightarrow v^{*2}\left(a+b\right)-
	\frac{b v^{*2}\gi^*}{g_l}>0.
	$$
	
	Substituting $c^*,\tilde{\lambda}^*$ into \eqref{eq:r-h-b1-full} and
	\eqref{eq:r-h-c1-full}  and 
	assuming $a\rightarrow 0^+$ gives
	\begin{align*}
		&\lim\limits_{a \rightarrow 0^+} b_1( g^*)= \frac{b
		\bar{\kappa} c^*}{b+v^{*2}}+ bv^{*^2}\left(1-\frac{\gi^*}{g_l}\right) 
		>0;\\
		&\lim\limits_{a \rightarrow 0^+} c_1( g^*)=  \bar{\kappa}
		c^*v^{*2} >0.
	\end{align*}
	Hence, all terms in the first column of \eqref{eq:r-h-table} are strictly positive. By the Routh-Hurwitz criterion $\left(g^*,\,\,\phi^*,\,\,\hat{g}^*\right)$ is a locally asymptotically stable equilibrium of \eqref{eq:scc-dynamics}.
	
	\item When $\pno<\pmax$, by Theorem \ref{thm:char-scc} there exist two equilibria $\left(g^*,\,\,\phi^*,\,\,\hat{g}^*\right)$ in the load satisfaction set $\mathcal{E}_s$ such that $\Delta P_{i}(g^*)=0$ for all $i\in N$. 
	We need to show that $\left\{\left(g^*,\,\,\phi^*,\,\,\hat{g}^*\right)\in\mathcal{E}_s:\gn^*<g_l \right\}$ is asymptotically stable. When $\left(g^*,\,\,\phi^*,\,\,\hat{g}^*\right)\in\mathcal{E}_s$ and $\gn^*<g_l$, then by 
	\eqref{eq:c(gt)}
	$$
		c^*= \left(\frac{ v^{* 2} }{\gn + g_l}\left(
		g_l - \gn^* \right) \right)^2 > 0\Rightarrow d_1 \left(g^*
		\right)>0.
	$$
	Notice that 
	\begin{align*}
		&\gi^* \leq \gn^*<g_l \overset{\eqref{eq:tilde-lambda}}{\Rightarrow}  
		\tilde{\lambda}^* >0, \,\, v^{* 2}b-\frac{2b v^{* 2}
		\gi^*}{\gn^*+g_l}>0.
	\end{align*}
	Substituting $c^*,\tilde{\lambda}^*$ into \eqref{eq:r-h-b1-full} and 
	\eqref{eq:r-h-c1-full} and assuming $a\rightarrow 
	0^+$ gives
	\begin{align*}
		&\lim\limits_{a \rightarrow 0^+} b_1( g^*)= \frac{b
		\bar{\kappa}c^*}{b+v^{* 2}}+ bv^{* 2}\left(1-\frac{\gi^*}{g_l}
		\right) >0;\\
		&\lim\limits_{a \rightarrow 0^+} c_1( g^*)=  \bar{\kappa}
		c^*v^{* 2} >0.
	\end{align*}
	Hence, all terms in the first column of \eqref{eq:r-h-table} are strictly positive. By the Routh-Hurwitz criterion $\left(g^*,\,\,\phi^*,\,\,\hat{g}^*\right)$ is a locally asymptotically stable equilibrium of \eqref{eq:scc-dynamics}. 
	\end{enumerate}\end{proof}
\subsection{Bifurcation Analysis} 

    In the proof of Theorem \ref{thm:char-scc} we showed that when $\pno<\pmax$ the system \eqref{eq:scc-dynamics} has two equilibria in the load satisfaction set $\mathcal{E}_s$  and one in the proportional allocation set $\mathcal{E}_p$, thus a total of three equilibria. When $\pno>\pmax$, the set $\mathcal{E}_s$ is empty and the system \eqref{eq:scc-dynamics} has a unique equilibrium in $\mathcal{E}_p$. In this section, we aim to understand the type of bifurcation our VCS controller undergoes at $\pno=\pmax$.
    
    To that end, we need to evaluate the stability of all equilibria of \eqref{eq:scc-dynamics} around $\pno=\pmax$. The load satisfaction equilibrium $\left\{\left(g^*,\,\,\phi^*,\,\,\hat{g}^*\right)\in\mathcal{E}_s:\gn^*>g_l\right\}$ results in lower voltage compared to  $\left\{\left(g^*,\,\,\phi^*,\,\,\hat{g}^*\right)\in\mathcal{E}_s:\gn^*<g_l\right\}$ studied in Section \ref{subsec:vcs-stability-of-efficient-points}. Moreover, when $\pno<\pmax$, the unique proportional allocation equilibrium $\left(g^*,\,\,\phi^*,\,\,\hat{g}^*\right)$ results in inefficient supply of flexible loads:
    \[
        P_i(g^*)=P_{0,i}+\frac{\kappa_i}{\bar{\kappa}}\left(\pmax-\pno\right)>P_{0,i}.
    \]
    Hence, understanding the type of bifurcation at $\pno=\pmax$ requires understanding the stability of the \textit{undesired} equilibria $\left\{\left(g^*,\,\,\phi^*,\,\,\hat{g}^*\right)\in\mathcal{E}_s:\gn^*>g_l\right\}$ and $\left(g^*,\,\,\phi^*,\,\,\hat{g}^*\right)\in\mathcal{E}_p$, when $\pno<\pmax$. The next Lemma shows that parameter $b$ affects the stability of the undesired load satisfaction equilibrium around $\gn^*=g_l$. 
    \begin{lemma}[Stabilization property of $b$ in \eqref{eq:scc-dynamics}]\label{lem:stab-reg}
    Let $b>0$ as in \eqref{eq:scc-dg-f}. Then,
    
        (1) For $b\in\left(0,\frac{E^2}{27}\right)$, there exist $0<m_b\leq M_b<+\infty$ such that all equilibria
    
        $$
        \left\{\left(g^*,\,\,\phi^*,\,\,\hat{g}^*\right)\in\mathcal{E}_s:\gn^*\leq g_l+m_b \right\}$$ 
        are locally asymptotically stable, and all equilibria
        
        $$
        \left\{\left(g^*,\,\,\phi^*,\,\,\hat{g}^*\right)\in\mathcal{E}_s:g_l+m_b\leq \gn^*\leq g_l+M_b\right\}$$ 
        are unstable.
        
        (2) For any $b_1,b_2\in\left(0,\frac{E^2}{27}\right)$, $b_1<b_2$, then
        $$
            m_{b_1}<m_{b_2} \qquad \text{and} \qquad M_{b_2}<M_{b_1}.
        $$
\end{lemma}
    
    We can now characterize the stability of all equilibria close to the bifurcation point $\pno=\pmax$.
    \begin{theorem}[Bifurcation Analysis of \eqref{eq:scc-dynamics}]\label{thm:bifurcation} 
    Consider the system \eqref{eq:scc-dynamics} with  $a>0$ satisfying \eqref{eq:stab-cond}. When $b\in\left(0,\frac{E^2}{27}\right)$ and for $a$ sufficiently small, the system \eqref{eq:scc-dynamics} undergoes a super-critical pitchfork bifurcation at $\pno=\pmax$. 
    \end{theorem}
    
    \begin{proof} When $\pno<\pmax$, we have shown in the proof of Theorem \ref{thm:char-scc} that the system \eqref{eq:scc-dynamics} has a total of three equilibria.
    When $\pno>\pmax$, the two equilibria in $\mathcal{E}_s$ disappear. Therefore, the system undergoes a pitchfork bifurcation. In Theorem \ref{thm:stab} we show that $\left\{\left(g^*,\,\,\phi^*,\,\,\hat{g}^*\right)\in\mathcal{E}_s:\gn^*\leq g_l\right\}$ when $\pno<\pmax$, and $\left(g^*,\,\,\phi^*,\,\,\hat{g}^*\right)\in\mathcal{E}_p$ when $\pno>\pmax$, are locally asymptotically stable. Moreover, Lemma \ref{lem:stab-reg} states that for $b\in\left(0,\frac{E^2}{27}\right)$, there exists $m_b>0$ such that any \textit{``undesired"} equilibrium $\left\{\left(g^*,\,\,\phi^*,\,\,\hat{g}^*\right)\in\mathcal{E}_s:g_l\leq\gn^*\leq g_l+m_b \right\}$ is locally asymptotically stable.
    Therefore, the system \eqref{eq:scc-dynamics} undergoes a super-critical pitchfork bifurcation if for $\pno<\pmax$, the remaining \textit{``undesired"} equilibrium $\left(g^*,\,\,\phi^*,\,\,\hat{g}^*\right)\in\mathcal{E}_p$ is unstable close to $\pmax$.

    When $\left(g^*,\,\,\phi^*,\,\,\hat{g}^*\right)\in\mathcal{E}_p$, then $\gn^*=g_l$ and by \eqref{eq:c(gt)}
	$$
		c^*=\left(\pno-\pmax\right)\frac{v^{*2}}{2g_l}<0.
	$$
	Then, by \eqref{eq:r-h-d1} $d_1(g^*)=\alpha_0(g^*)=\bar{\kappa}ac^*v^{* 2}
	<0$. Therefore, there exists at least one sign change 
	between the terms in the first column of \eqref{eq:r-h-table}. By the 
	Routh-Hurwitz criterion $g^*$ is an unstable equilibrium and the system \eqref{eq:scc-dynamics} undergoes a super-critical pitchfork bifurcation at $\pno=\pmax$.  
	\end{proof}
	
Theorem \ref{thm:bifurcation} implies that the VCS controller succeeds in transforming the type of bifurcation at $\pno=\pmax$. The transformation of the bifurcation comes at the cost of stabilizing the undesired equilibrium $\left\{\left(g^*,\,\,\phi^*,\,\,\hat{g}^*\right)\in\mathcal{E}_s:\gn^*\geq g_l\right\}$ in some region $\{(g^*,\phi^*,\hat{g}^*)\in\mathcal{E}_s: g_l\leq \gn^*<g_l+m_b\}$. 
\begin{remark}\label{rem:a-b}
	Theorem \ref{thm:stab} shows that parameter $a>0$ in \eqref{eq:scc-dghat} affects the stability of the efficient equilibria, while Theorem \ref{lem:stab-reg} shows that parameter $b>0$ in \eqref{eq:scc-dg-f} affects the stability of the undesired equilibrium $\left\{\left(g^*,\,\,\phi^*,\,\,\hat{g}^*\right)\in\mathcal{E}_s:\gn^*\geq g_l\right\}$. Choosing $a\rightarrow 0^{+}$ stabilizes the efficient equilibria, but there is not a single desired value for $b\in\left(0,\frac{E^2}{27}\right)$. Smaller values of $b\in\left(0,\frac{E^2}{27}\right)$ lead to a smaller $m_b>0$ in Theorem \ref{lem:stab-reg} and thus a smaller region around $\gn^*=g_l$ where the undesired equilibrium is locally asymptotically stable. However, $b=0$ would result in $b_1=0$ in \eqref{eq:r-h-b1}. In that case, the linearization would fail and our analysis would be inconclusive. 
\end{remark}
\begin{figure}
    \begin{center}
        \includegraphics[width=1.0\columnwidth]{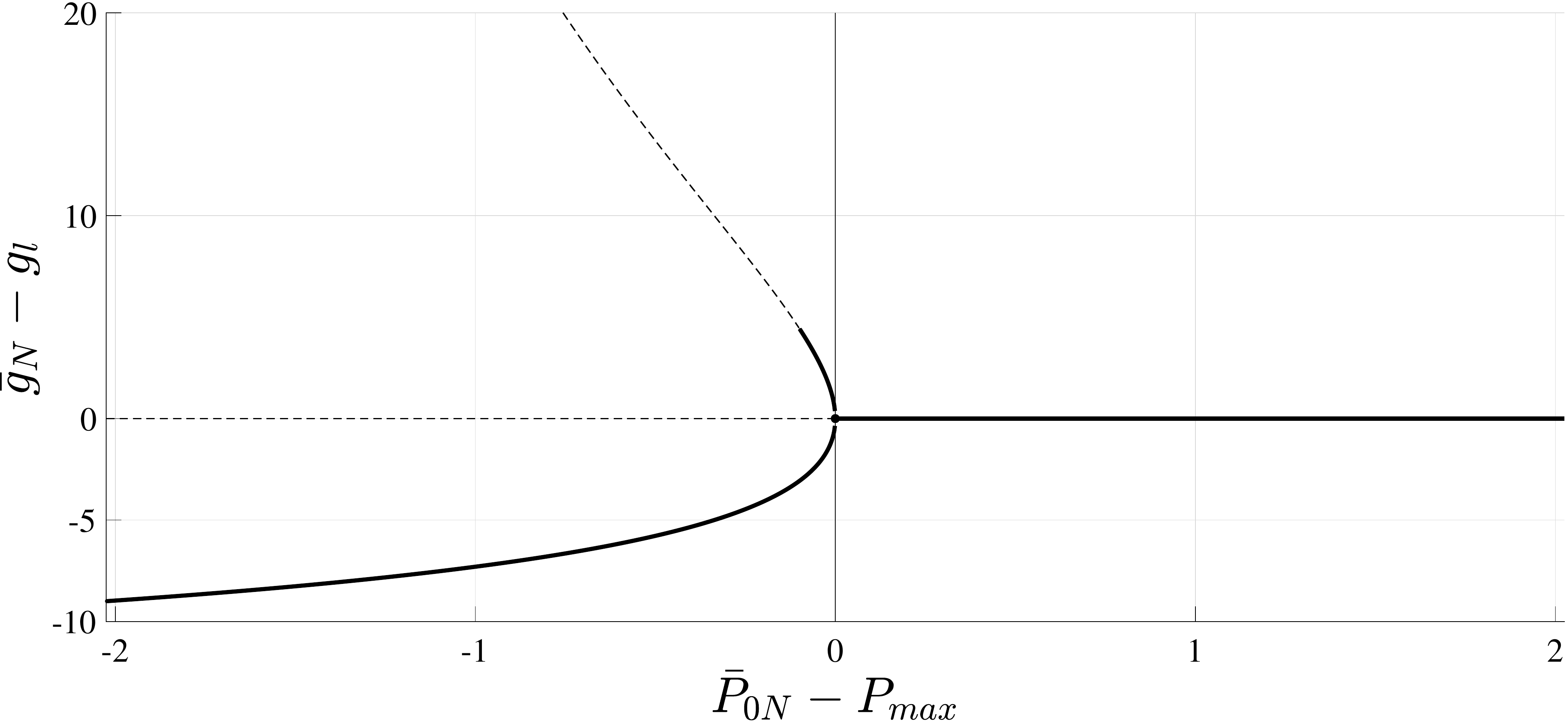}      
    \caption{Bifurcation Diagram of \eqref{eq:scc-dynamics}. As total demand approaches the capacity of the line ($\pno=\pmax$), an unstable equilibrium becomes locally asymptotically stable. At the bifurcation point, the two locally asymptotically stable equilibria and the unstable equilibrium coalesce, resulting in a single locally asymptotically stable equilibrium for $\pno>\pmax$.  }                    
    \label{fig:supercriticalpitchfork}                                                         
    \end{center}                                                            
\end{figure}
\section{Numerical Results}\label{sec:results}
In this section, we validate our theoretical results using numerical 
illustrations. In all the experiments we start the simulations with initial set-points 
such that  $\pno<P_{\text{max}}$ and with conductances close to the 
equilibrium $g^*$ where all demands are met, i.e., $g^*\in \mathcal{E}_s\cap M$. 
We explore the parameter space by  \textit{slowly} varying the demand ($P_0$) 
with time and observing the changes in the equilibria. 

The load-shedding property of the VCS controller is better understood using more than one loads. For that, we will study a system where loads $1$ and $2$ are flexible and load $3$ is inflexible. First, we look at the behavior of the system when the conventional controller \eqref{eq:infl} is applied to all three loads. Figure \ref{fig:power-dec} shows that when total demand exceeds the maximum transferable power through the line, voltage collapses.

\begin{figure}
    \begin{center}
     	\includegraphics[width=1.0\columnwidth]{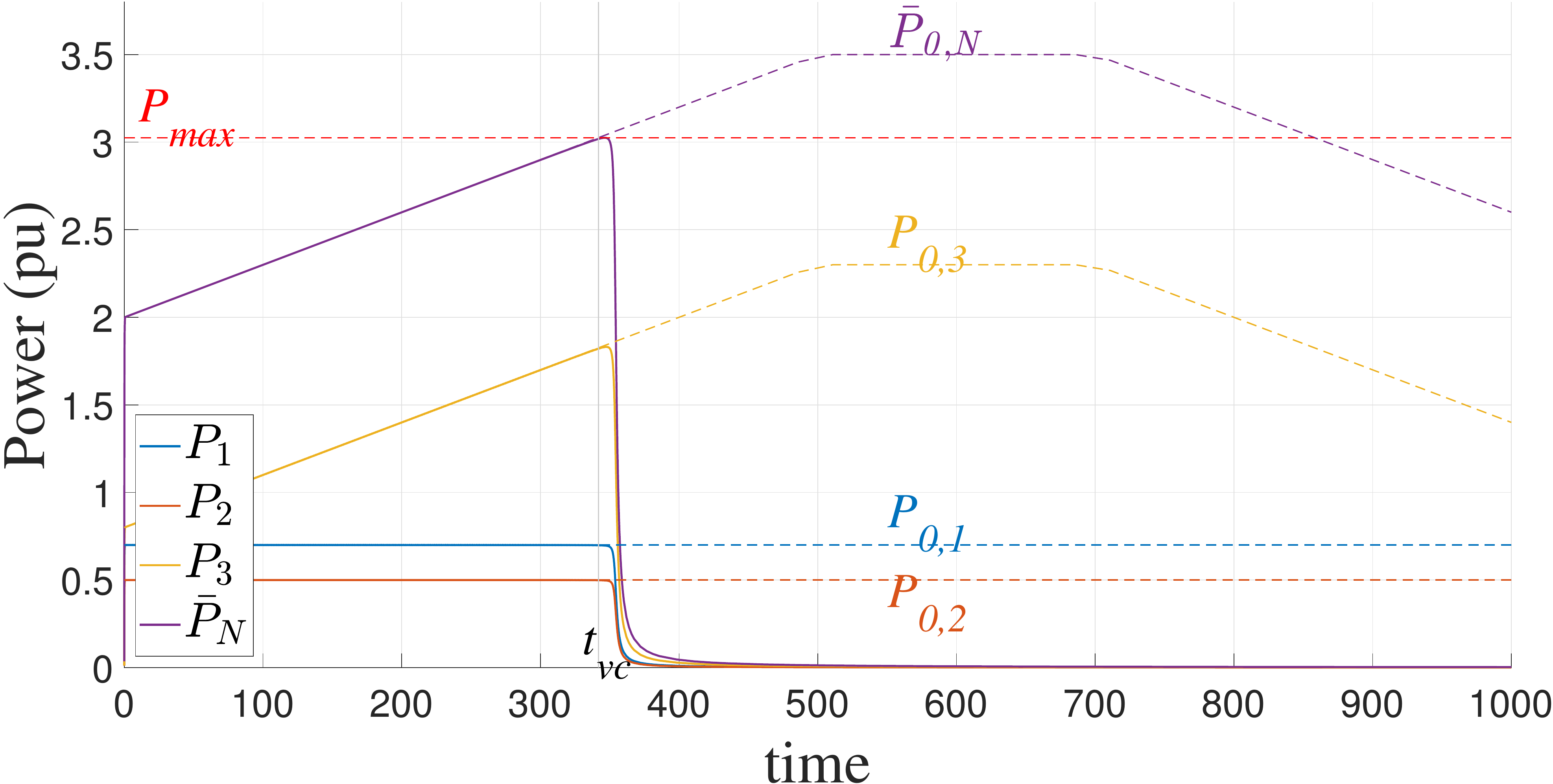}
		\caption{Power supply when all three loads are controlled with the traditional control strategy. When $\pno>\pmax$ voltage collapses.}
	\label{fig:power-dec}                                                    
\end{center}                           
\end{figure}
On the other hand, Figure \ref{fig:power-sec} shows how flexible loads adjust in the overloading regime, when the flexible loads are controlled using the VCS controller. The adjustment is proportional to the parameter $\kappa_i$, consistent with the definition of $\mathcal{E}_p$ in Section \ref{sec:eqlbrm-analysis}. Finally,  Figure \ref{fig:vcs-collapse} illustrates how voltage inevitably collapses when the demand of inflexible loads exceed the capabilities of the line.
\begin{figure}
    \begin{center}
    	\includegraphics[width=1.0\columnwidth]{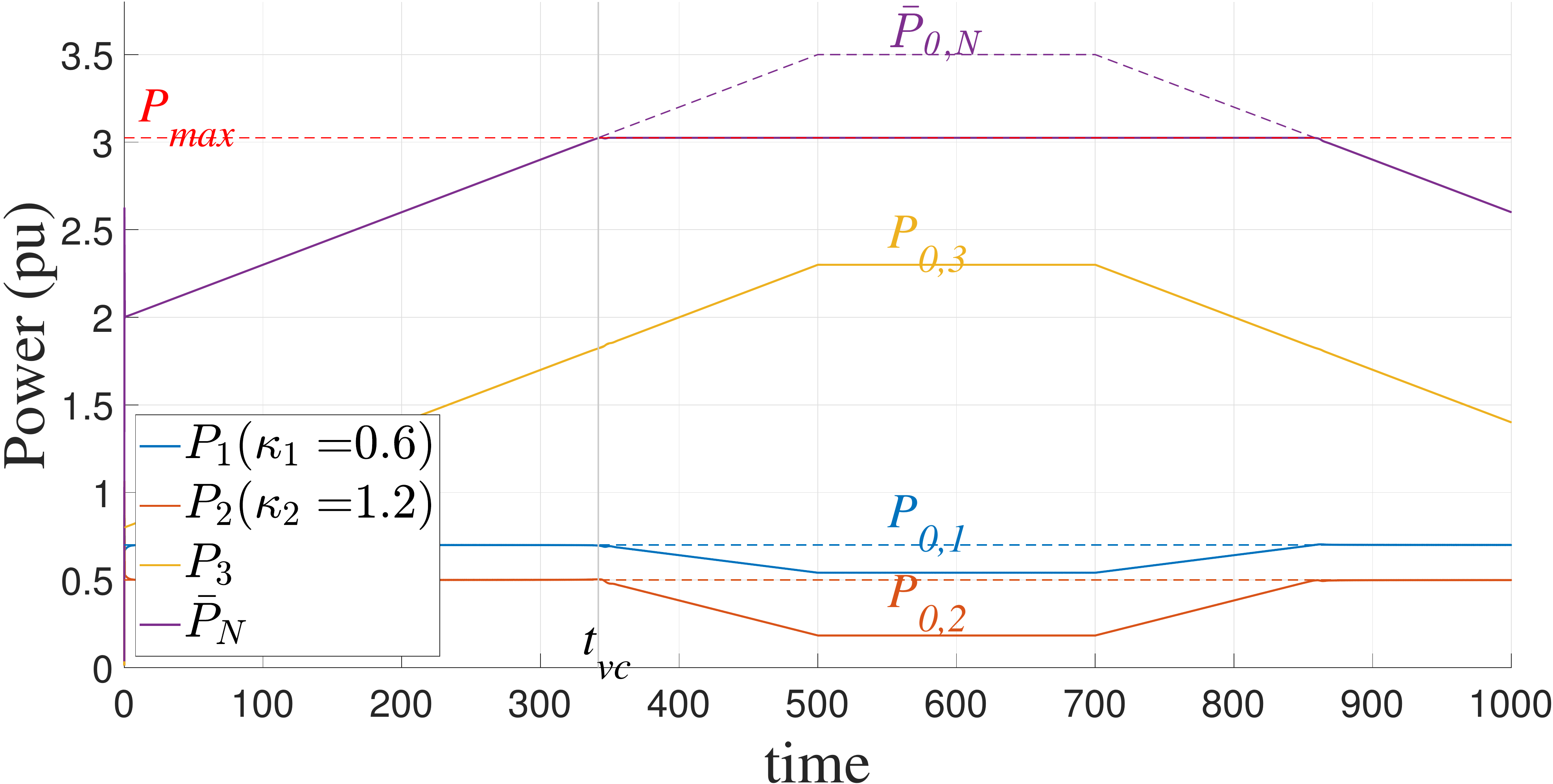}
		\caption{Power supply for $a=0.1$, $b=0.1$ when
		loads $1,2$ are flexible and load $3$ is inflexible. When total demand exceeds the maximum, the adjustment of power supply to the flexible loads is proportional to the parameters $\kappa_i$ $\forall i=1,2$.}
		
	\label{fig:power-sec}                       
    \end{center}                               
\end{figure}
\begin{figure}
    \begin{center}
        \includegraphics[width=1.0\columnwidth]{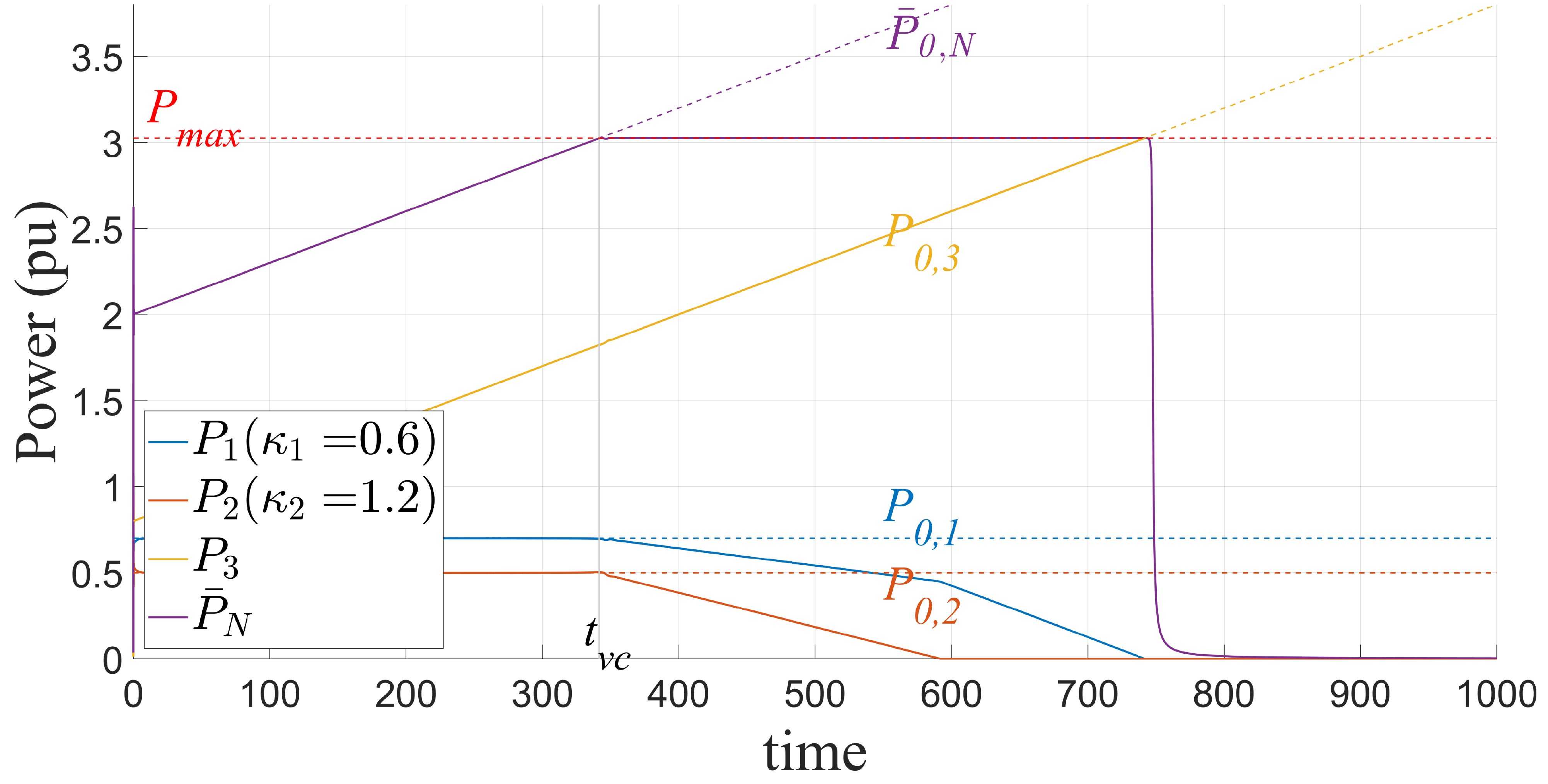}      
    \caption{Power supply for $a=0.1$, $b=0.1$ when 
		loads $1,2$ are flexible and load $3$ is inflexible. Voltage 
		collapses when the demand of the inflexible load exceeds the 
		capabilities of the line.}
	\label{fig:vcs-collapse}                                
    \end{center}       
\end{figure}

Remark \ref{rem:a-b} highlights the impact of parameters $a,b$ on the stability of the equilibria of \eqref{eq:scc-dynamics}. To better understand the impact of $a$ and $b$, we will consider a simple version of the DC network in Figure \ref{fig:nbs} with one flexible and one inflexible load. Figure \ref{fig:stab-powers-sec} verifies that small, non-zero values of both $a$ and $b$ result in stable power supply. On the other hand, when $a$ is bigger and $\pno>\pmax$, the equilibrium that ensures proportional shedding becomes unstable. Moreover, when total demand is smaller than $\pmax$ but very close, larger values of $b$ result in the undesired equilibrium $g^*\in\mathcal{E}_s\cap M^c$ being stable. Finally, Figure \ref{fig:stab-conds-sec} reveals that for small values of $a$ and total demand close to $\pmax$, we always track the desired equilibrium $g^*\in\mathcal{E}_s\cap M$. 
\begin{figure}
    \begin{center}
        \includegraphics[width=1.0\columnwidth]{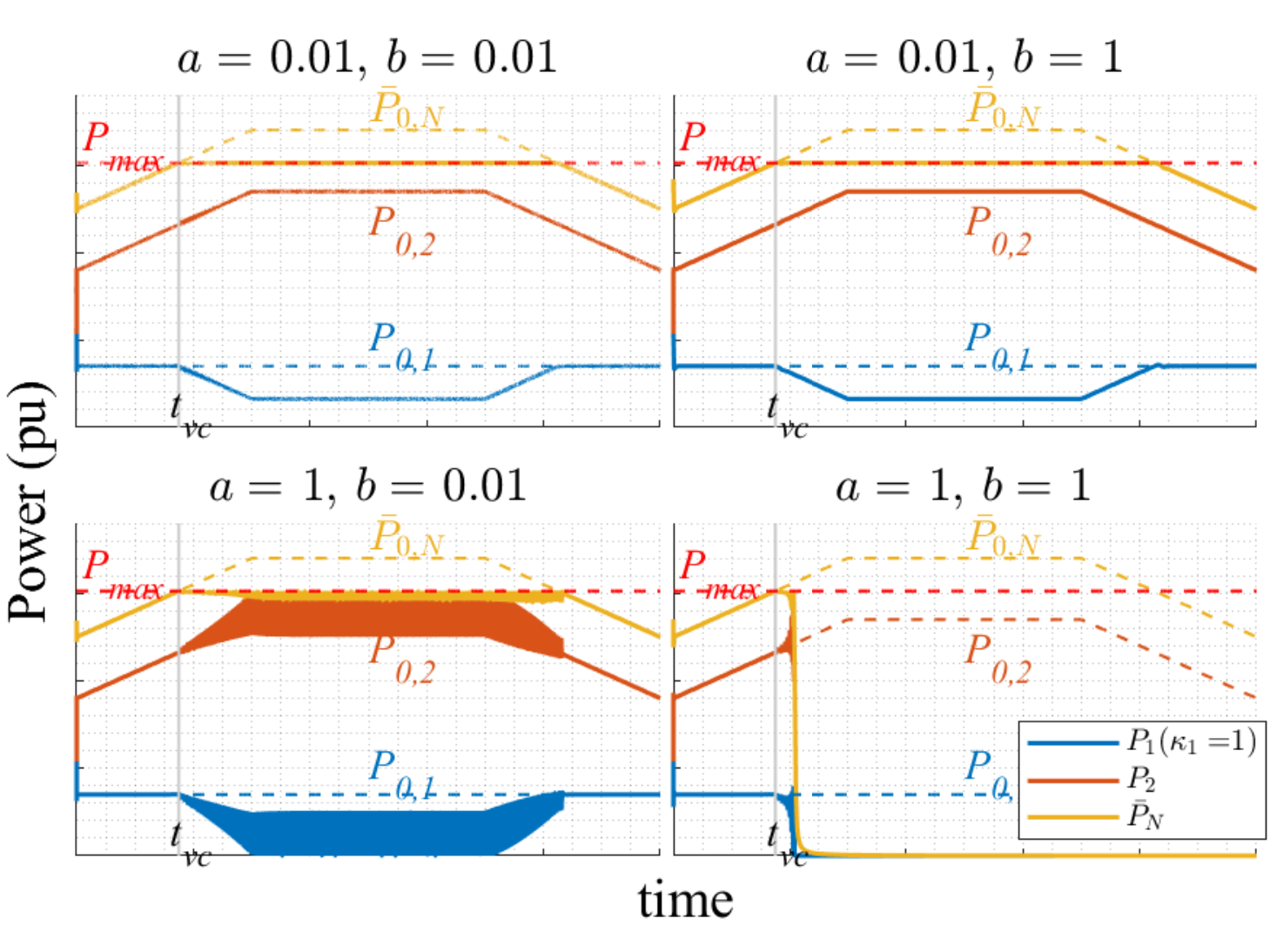}      
 	\caption{Power supply for different values of $a,b$ when load $1$ is flexible and load $2$ is inflexible. For $a$ small --upper two plots-- and $\pno>\pmax$, power supply to load $1$ adjusts accordingly. For higher $a$ --lower two plots-- and $\pno>\pmax$, voltage collapses.   
	}                    
    \label{fig:stab-powers-sec}                                                    
    \end{center}                                                            
\end{figure}
\begin{figure}
    \begin{center}
    \includegraphics[width=1.0\columnwidth]{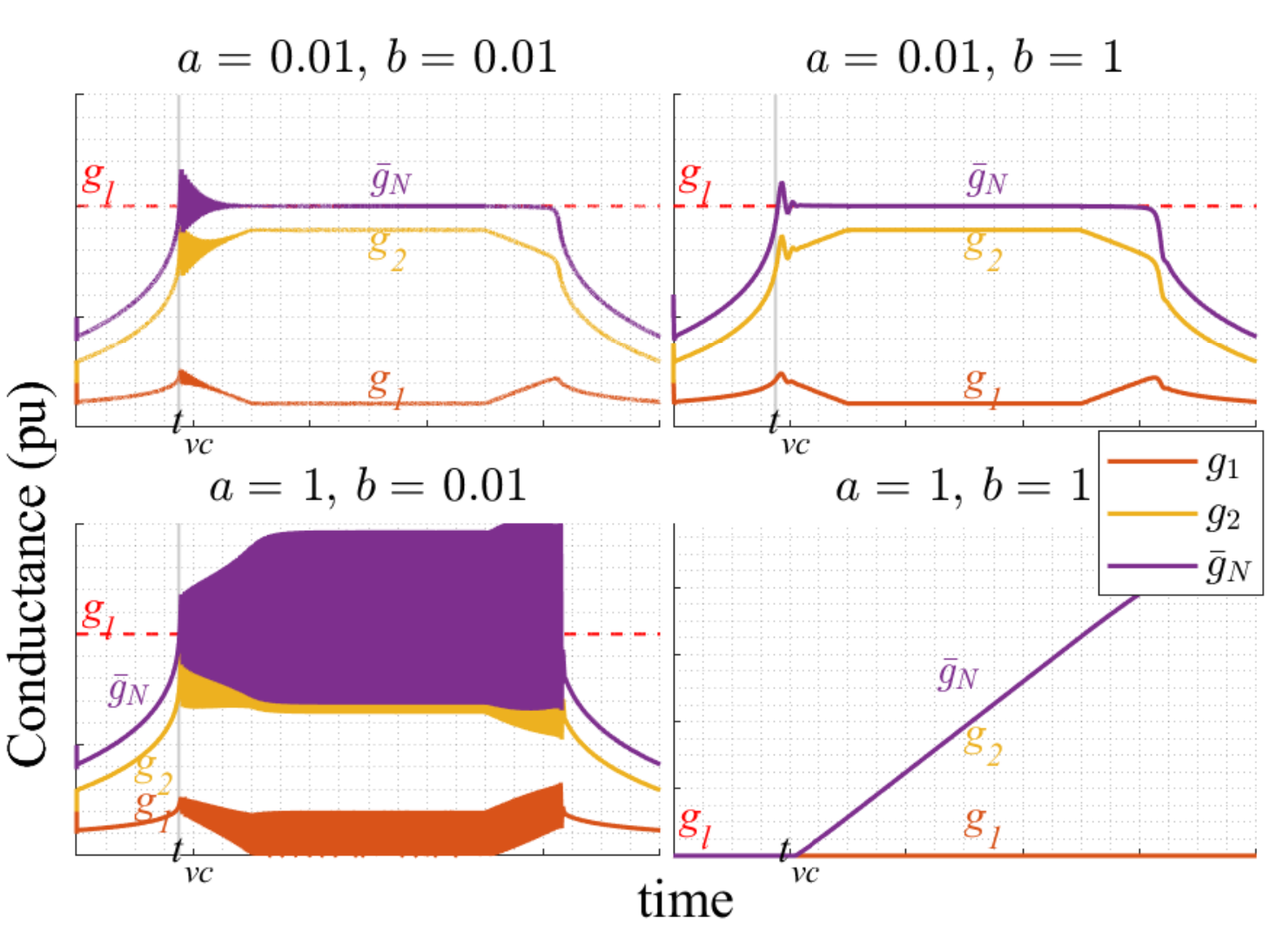}      
    
 	\caption{Response of conductances for different values of $a,b$ when load $1$ is flexible and load $2$ is inflexible. For small $a$, when $\pno<\pmax$, the VCS controller tracks $g^*\in\mathcal{E}_s\cap M$. When $\pno>\pmax$ then $\gn^*=g_l$. For higher values of $a$, conductances diverge when $\pno>\pmax$. }
	\label{fig:stab-conds-sec}                                                    
    \end{center}                                                            
\end{figure}
\section{Conclusions}\label{sec:conclusions}
This work seeks to initiate the study of voltage collapse stabilization as a 
mechanism to allow for more efficient and reliable operation of electric power 
grids. 
We develop a Voltage Collapse 
Stabilization controller that is able to not only prevent voltage collapse, but 
also proportionally distribute the curtailment among all flexible loads in a star DC network. The methodology can be readily applied to a fully reactive alternating current star network, where the controller changes the suscenptance instead of the conductance.   
Further research needs to be conducted to fully characterize the behavior of our 
solution at the bifurcation point. The point where $\pno=P_{\text{max}}$ is a 
non-trivial point, where the Jacobian of the system is identically zero and thus 
requires the treatment of higher order dynamics. We identify two desired 
extensions of this work that are the subject of current research: (a) extending the 
analysis to a general DC network and (b) extending the analysis to a general AC 
network.


\begin{ack}                                 
This work was supported by ARO through contract W911NF-17-1-0092, US DoE EERE award DE-EE0008006, and NSF through grants CNS 1544771, EPCN 1711188, AMPS 1736448, and CAREER 1752362. The authors are also grateful for the contribution of Dr. Fernando Paganini, Aurik Sarker, and Jesse Rines to an earlier version of this work.                  
\end{ack}




\appendix
\section{Proofs of Lemmas}\label{appdx:proof-of-lemmas}

\paragraph*{Lemma \ref{lem:evals-r1+identity}.}
\begin{proof}
	An eigenvalue-eigenvector pair $\left(\rho,x\right)\in\C\times\C^n$ of $B$ needs to satisfy
	\begin{equation}\label{eq:B-conds}
		B x = \rho x
	\end{equation} 
	Since $w\mathbf{1}_n^T$ is a rank one matrix, it has $n-1$ eigenvalues $\lambda_i\left(w\mathbf{1}_n^T\right)=0$, $i=1,\dots,n-1$, and one non-zero eigenvalue $\lambda_1\left(w\mathbf{1}^T\right)=\mathbf{1}_n^T w$. Moreover, $q\I$ is a scaled identity matrix and thus shifts the eigenvalues of $w \mathbf{1}_n^T$ by $q$, {i.e.,} $\rho_i\left(B\right)=\lambda_i\left(w \mathbf{1}_n^T\right)+q$. Finally, we can check that $\left\{e_1-e_2,e_1-e_3,\dots,e_1-e_n,w\right\}$ spans $\C^n$ and every pair of \eqref{eq:evecs-evals-r1+identity} satisfies \eqref{eq:B-conds}.
\end{proof}

\paragraph*{Lemma \ref{lem:jsc-evals}.}
\begin{proof}
	If $(\lambda,u)\in\C\times\C^{n+1+n_F}$ an eigenvalue-eigenvector pair of $J_{SC}(g)$, then $J_{SC}(g)u=\lambda u$ with $u=\left[\begin{array}{cccc} u_F & u_I & u_{\phi} & \hat{u} \end{array}\right]^T\in\R^{n_F}\times\R^{n_I}\times\R\times\R^{n_F}$, which further leads to
\begin{subequations}\label{eq:jsc-conds}
	\begin{align}
		\label{eq:jsc-conds-a-1}& \frac{2v^2}{\gn+g_l}g_i
		\bar{u}_N  - \left(v^2+\lambda+b\right)u_{F_i} 
		+ b
		\hat{u}_i = \kappa_i u_{\phi},&& \forall i\in F \\
		\label{eq:jsc-conds-a-2}& \frac{2v^2}{\gn+g_l}g_i \bar{u}_N - 
		\left(v^2+\lambda\right) u_{I_i} =  0,&& \forall i\in I\\
		\label{eq:jsc-conds-b}& c\cdot\bar{u}_N=\lambda u_{\phi}& \\
		\label{eq:jsc-conds-c}& au_F-a\hat{u}=\lambda \hat{u}& 
	\end{align}
\end{subequations}
with 
\begin{equation}\label{eq:u-bar-g}
	\bar{u}_{F} =\sum_{i\in F} u_{F_i}, \quad \bar{u}_{I} =\sum_{i\in I} u_{I_i}, \quad \bar{u}_N=\bar{u}_F+\bar{u}_I.
\end{equation}
First we will prove that $-a$ is an 
eigenvalue of $J_{SC}\left(g\right)$ only when 
\begin{equation}\label{eq:degen-a-cond}
	a=v^2-\frac{2v^2}{\gn+g_l}\sum_{i\in I} g_i.
\end{equation}
If $\lambda=-a$, by \eqref{eq:jsc-conds-c} $u_{F}=
\mathbf{0}_{n_F}$ and by \eqref{eq:u-bar-g}, $\bar{u}_N=
\sum_{i\in I} u_i$. Summing \eqref{eq:jsc-conds-a-2} yields
$$
	\left(\frac{2v^2}{\gn+g_l} \sum_{i\in I} g_i - v^2+a\right)\bar{u}_N=0.
$$
If \eqref{eq:degen-a-cond} is true for some $g\in\R^n_{\geq 0}$, then for any 
$\bar{u}_N\neq 0$, by \eqref{eq:jsc-conds-b}, we get 
\begin{align*}
		&\eqref{eq:jsc-conds-b} \Rightarrow u_\phi= \frac{c\cdot\bar{u}_N}
		{\lambda}=-\frac{c\cdot\bar{u}_N}{a};\\
		&\eqref{eq:jsc-conds-a-2} \overset{\eqref{eq:degen-a-cond}}{\Rightarrow} u_i=\frac{g_i 
		 }
		{ \gi }\bar{u}_N, \quad \forall i
		\in I; \\
		&\eqref{eq:jsc-conds-a-1} \Rightarrow \hat{u}_i=-\frac{1}{b}\left( 
		\frac{2v^2g_i}{\gn+g_l}+\frac{\kappa_i c}{a}\right)\bar{u}_N, \, 
			\forall i
		\in F. 
\end{align*}
The respective eigenvector is
$$
	\left[\begin{array}{cccc}
		\mathbf{0}_{n_F} & \frac{1 }
		{\gi}g_I  & -\frac{c}{a} & 
		- 
		\frac{1}{b}\left(\frac{2v^2}{\gn+g_l}g_F+\frac{c}{a}\kappa
	\right)\end{array}\right]^T \bar{u}_N,
$$
for some $\bar{u}_N \neq 0$ and corresponding eigenvalue $-a$. 

If  \eqref{eq:degen-a-cond} does not hold, then $-a$ is an eigenvalue if $\bar{u}_N=0$. By contradiction, substituting $\bar{u}_N = 0$ and $\lambda=-a\neq 0$ into 
\eqref{eq:jsc-conds-b} yields 
$u_{\phi}=0$. Substituting $u_{\phi}=\bar{u}_N=u_{F_i}=0$ for all $i\in N$ into \eqref{eq:jsc-conds-a-1} yields $\hat{u}=\mathbf{0}_{n_F}$. Substituting $\bar{u}_N=0$ into \eqref{eq:jsc-conds-a-2} yields
$$
    \left(v^2 -a\right)u_{I_i}=0\overset{\eqref{eq:stab-cond}}{\Rightarrow} u_{I_i}=0 \quad \forall i\in I.
$$
Therefore, $u=\mathbf{0}_{n+1+n_F}$ and $-a$ cannot be an eigenvalue of $J_{SC}(g)$. 

When $\lambda\neq -a$, then by \eqref{eq:jsc-conds-c}
\begin{equation}\label{eq:jsc-conds-c-new}
	\hat{u}=\frac{a}{\lambda+a}u_F.
\end{equation}
By \eqref{eq:stab-cond}, $g_i>0$ $\forall i\in I$, therefore $\gi=\sum_{i\in I} g_i>0$.  Summing  \eqref{eq:jsc-conds-a-2} over 
all $i\in I$ gives
$$
\frac{2v^2}{\gn+g_l} \gi \bar{u}_F +\left(\frac{2v^2}{\gn+g_l} 
\gi-v^2-\lambda\right)\bar{u}_I=0 
$$	
\begin{equation}\label{eq:bar-uF}
\begin{split}
\overset{\left(\gi>0\right)}{\Rightarrow} \bar{u}_F=&\left(\frac{\gn+g_l}
{2v^2\gi}\left(v^2+\lambda\right)-1\right)\bar{u}_I \\
=&f_{\bar{u}_F}\left(\lambda\right)\bar{u}_I.	
\end{split}	
\end{equation}
Substituting \eqref{eq:jsc-conds-c-new} into \eqref{eq:jsc-conds-a-1} and summing over all $i\in F$ gives
\begin{flalign*}
\left(\frac{2v^2\gf}{\gn+g_l}-v^2 -\lambda-  b +
 \frac{ab}{\lambda+a} \right)\bar{u}_F +\frac{2v^2\gf}
{\gn+g_l}\bar{u}_I 
\end{flalign*}
\begin{equation}\label{eq:bar-uphi}
\begin{split}
\overset{\eqref{eq:bar-uF}}{=}& \left[\left( \frac{2v^2\gf}{\gn+g_l}- v^2 
-\lambda - b + \frac{ab}{\lambda+a} \right) f_{\bar{u}_F}\left(\lambda\right)
 \right. \\
&\quad\left.+\frac{2v^2\gf}{\gn+g_l}\right]\bar{u}_I \\
=& f_{\bar{u}_I}\left(\lambda\right)\bar{u}_I=\bar{\kappa}u_\phi.
\end{split}
\end{equation}
Multiplying \eqref{eq:jsc-conds-b} by $\bar{\kappa}$ yields	
\begin{align*}
	\lambda \bar{\kappa}u_\phi = \bar{\kappa} c \left(\bar{u}_F+ \bar{u}
	_I\right) \overset{\eqref{eq:bar-uF}}{=}  \frac{\bar{\kappa}c
	\left(\gn+g_l\right)}{2v^2\gi}
	\left(v^2+\lambda\right) \bar{u}_I.
\end{align*}
Plugging \eqref{eq:bar-uphi} above yields
\begin{equation}\label{eq:jsc-evals-eq}
	\left[\lambda  f_{\bar{u}_I}\left(\lambda\right) -\bar{\kappa} c \frac{\gn
	+g_l}{2v^2\gi}\left(v^2+\lambda\right) \right] 
	\bar{u}_I = 0.
\end{equation}
Equation \eqref{eq:jsc-evals-eq} implies that either the first term or $\bar{u}_I$ is zero. We will distinguish between the two cases:
\begin{itemize}[leftmargin=*]
	\item Let $\bar{u}_I=0$. Substituting $\bar{u}_I=0$ into \eqref{eq:bar-uF} yields $\bar{u}_F=0$ and thus $\bar{u}_N=\bar{u}_I+\bar{u}_F=0$. By substituting \eqref{eq:jsc-conds-c-new} and $\bar{u}_F=\bar{u}_N=0$ into \eqref{eq:jsc-conds-a-1} and summing over all $i\in F$ we get $u_\phi=0$. 
Moreover, substituting $\bar{u}_I=\bar{u}_F=0$ into \eqref{eq:jsc-conds-a-2} results in
	\begin{equation}\label{eq:j_I-evals-eq}
			\mathbf{0}_{n_I}=-\left(v^2 +\lambda\right)u_I 
	\end{equation}
	Equation \eqref{eq:j_I-evals-eq} can be satisfied either when $u_I=
	\mathbf{0}_{n_I}$ or $\lambda=-v^2$. If $u_{I}=\mathbf{0}_{n_I}$ in \eqref{eq:j_I-evals-eq}, by 
	substituting $u_\phi=0$ and \eqref{eq:jsc-conds-c-new} into 
	\eqref{eq:jsc-conds-a-1} we get
	\begin{align}\label{eq:jsc-pre-cond1}
		&\left(v^2+\lambda+b-\frac{ab}{\lambda+a}\right)u_F=\mathbf{0}_{n_F}. 
	\end{align}
	If $u_F= \mathbf{0}_{n_F}$ then by \eqref{eq:jsc-conds-c} $\hat{u}=\mathbf{0}_{n_F}$ and $u=\mathbf{0}_{n+1+n_F}$ cannot be an eigenvector of \eqref{eq:jacob-scc}. Therefore, $u_F\neq \mathbf{0}_{n_F}$ which means that the first term in \eqref{eq:jsc-pre-cond1} is zero. Multiplying the first term in \eqref{eq:jsc-pre-cond1} by $\lambda+a$ and rearranging the terms results in \eqref{eq:jacob-scc-cond1}.
	Equation \eqref{eq:jacob-scc-cond1} is a second order polynomial with 
	respect to $\lambda$ whose discriminant is
	\begin{align*}
		\Delta =& \left(a+b+v^2\right)^2-4av^2 \\
		 	   =& \left(a-v^2\right)^2 +b^2+2ab+2bv^2>0.
	\end{align*}
	Therefore, it has two distinct solutions $\lambda_1,\lambda_2$.	We can check that there exist $2n_F-2$ eigenvalue-eigenvector pairs of $J_{SC}\left(g\right)$ of the form 
    	\begin{align*}
        	&\left(\lambda_1,\left[\begin{array}{cccc} e_1-e_i & \mathbf{0}_{n_I} & 0 &
        	\frac{a\left(e_1-e_i\right)}{\lambda_1+a} \end{array}\right]^T\right),& i=2,\dots n_F; \\
        	&\left(\lambda_2,\left[\begin{array}{cccc} e_1-e_i & \mathbf{0}_{n_I} & 0 &
        	\frac{a\left(e_1-e_i\right)}{\lambda_2+a} \end{array}\right]^T\right),& i=2,\dots n_F; 
    	\end{align*}
	that satisfy \eqref{eq:jsc-conds}. 	In fact, there exist \textit{exactly} $2n_F-2$ eigenvalues of this form. By \eqref{eq:jsc-conds-a-1}
	$$
	\underbrace{\left(\frac{2v^2}{\gn+g_l}g_F\mathbf{1}_{n_F}^T -  
	v^2\I_{n_F}\!\right)}_{J_F(g)}u_F\!=\!\left( \! b-\frac{ab}{\lambda
	+a}+\lambda\! \right) \!  u_F,
	$$
	{i.e.,} $\left(\left(b-\frac{ab}{\lambda+a}+\lambda\right),u_F\right)$ is an eigenvalue-eigenvector pair of $J_F(g)$. Matrix $J_F(g)$ is a $n_F\times n_F$ RPSI matrix. By Lemma \ref{lem:evals-r1+identity}, there exist exactly $n_F-1$ eigenvalue-eigenvector pairs for each $\lambda_i$, $i\in\{1,2\}$, that satisfy $\bar{u}_F=0$. Therefore, $2\left(n_F-1\right)$ in total. \\
	
	If $\lambda=-v^2$ in \eqref{eq:j_I-evals-eq}, then by substituting $u_\phi=0$ 
	and \eqref{eq:jsc-conds-c-new} into \eqref{eq:jsc-conds-a-1} gives
	$$
	\left(b+\frac{ab}{-v^2 +a}\right)u_F=\mathbf{0}_{n_F} 
	\overset{\left(\lambda\neq-a\right)}{\Rightarrow} u_F=\mathbf{0}_{n_F}.
	$$
	When $u_F=\mathbf{0}_{n_F}$, by \eqref{eq:jsc-conds-c} $\hat{u}=\mathbf{0}
	_{n_F}$. Since $\bar{u}_F=0$, we can rewrite \eqref{eq:j_I-evals-eq} as
	$$
	\underbrace{\left(\frac{2v^2}{\gn+g_l}g_I \mathbf{1}_{n_I}^T-v^2\I_{n_I}
	\right)}_{J_I(g)}u_I=\lambda u_I,
	$$	
	{i.e.,} $\left(\lambda,u_I\right)$ is an eigenvalue-eigenvector pair of $J_I(g)$ that satisfies $\bar{u}_N=0$. Matrix $J_I(g)$ is a $n_I\times n_I$ RPSI matrix. By Lemma \ref{lem:evals-r1+identity} there exist exactly $n_I-1$ eigenvalue-eigenvector pairs that satisfy $\bar{u}_I=0$ and these are
	$$
	\left(-v^2,\left[\begin{array}{cccc} \mathbf{0}_{n_I} & e_{1}-e_i & 0 &
	\mathbf{0}_{n_F} \end{array}\right]\right), \quad  i = 1,\dots,n_I-1.
	$$
	\item If $\bar{u}_I\neq 0$, then by \eqref{eq:jsc-evals-eq}
	\begin{align*}
		&\lambda f_{\bar{u}_I}\left(\lambda\right) - \bar{\kappa}c \frac{\gn
		+g_l}{2v^2\gi}\left(v^2+\lambda\right)=0 \\
		\Rightarrow&- \lambda^3+\lambda^2\left(-\tilde{\lambda}
		\left(\gn\right)-v^2-b\right)+\\
		&+\lambda\left(-v^2\left(\tilde{\lambda}\left( \gn
		\right)+b\right)-\bar{\kappa}c+ \frac{2bv^2\gi}{\gn+
		g_l}
		\right)-\bar{\kappa}cv^2 \\
		&+\frac{ab\lambda}{\lambda+a}\left( v^2+\lambda- 
		\frac{2v^2\gi}{\gn+g_l} \right)=0.
	\end{align*}

	We multiply both sides with $-\left(\lambda+a\right)$ and group terms together to obtain \eqref{eq:jacob-scc-cond2}.
    \end{itemize}
    Finally, the computation of \eqref{eq:jacob-scc-cond2-main} is based on the assumption $\lambda\neq-a$, which we have shown is true for all $g\in\R^n$ such that $a\neq v^2-\frac{2v^2\gn}{\gn+g_l}$. When $a=v^2-\frac{2v^2\gn}{\gn+g_l}$, then 
	$\lambda=-a$ is an eigenvalue of $J_{SC}(g)$. However, if we substitute $\lambda=-a$ into \eqref{eq:jacob-scc-cond1} we find that
	\begin{align*}
	    (-a)^2+(a+b+v^2)(-a)+av^2=ab\neq 0,   
	\end{align*}
	{i.e.,} $\lambda=-a$ is not a root of \eqref{eq:jacob-scc-cond1}. Therefore, when $a=v^2-\frac{2v^2\gn}{\gn+g_l}$, the statement of Lemma \ref{lem:jsc-evals} is consistent only if $\lambda=-a$ is a root of \eqref{eq:jacob-scc-cond2-main}.
	
	\begin{align*}
		&\bar{\kappa} acv^2-a\left( a v^2\tilde{\lambda} +\bar{\kappa} c
		\left(a+v^2\right)
		\right)+ \\
		&+a^2\left(\bar{\kappa} c
		 +v^2\left( a + b  \right) - \frac{2bv^2 \gi}{\gn+g_l} +  
		 \tilde{\lambda}\left(
		a + v^2\right)\right)-			\\
		&- a^3 \left( b + a + v^2 + \tilde{\lambda}\right) + a^4 \\
	   =&a^2b\underbrace{\left(v^2 -\frac{2v^2 
		 \gi}{\gn + g_l} \right)}_{a} -  a^3 b=0,
	\end{align*}
	In the special case where $g\in\R^n$ is such that $a=v^2-\frac{2v^2\gn}{\gn+g_l}$, then the derivation of \eqref{eq:jacob-scc-cond2-main} discards the root $\lambda=-a=-v^2+\frac{2v^2\gn}{\gn+g_l}$. Therefore, the statement of the lemma is consistent for all $g\in\R^n$. 
    \end{proof}

\paragraph*{Lemma \ref{lem:stab-reg}.}
\begin{proof}
    (1) We start by observing that the only term in the first column of the Ruth-Hurwitz table that is affected by $b>0$ is $b_1$.  	Taking $a\rightarrow 0^+$ in \eqref{eq:r-h-b1-full} yields
	\begin{align}
		\lim\limits_{a\rightarrow 0^+} b_1\left(g^*\right)=& \left(b+\tilde{\lambda}^*\right) 
		\left(v^{* 2}+\frac{\bar{\kappa} c^*}{ b + 
		\tilde{\lambda}^*+ v^{* 2}}\right)-\\
		&-\frac{2b v^{* 2}\gi^*}{\gn^*+g_l} \\
		=& f(\gn^*)-\frac{2b v^{* 2}\gi^*}{\gn^*+
		g_l}.\label{eq:r-h-b1-cond}
	\end{align}
 
    We will study the sign of all terms of $b_1$ in \eqref{eq:r-h-b1-cond} for $\left\{\left(g^*,\,\,\phi^*,\,\,\hat{g}^*\right)\in\mathcal{E}_s: \gn^*\geq g_l\right\}$ and for different values of $b>0$.
    
    Notice that $
		\gn^*>g_l \overset{\eqref{eq:tilde-lambda}}{\Rightarrow} 
		\tilde{\lambda}^*< 0$ and
    $$
        \frac{d\tilde{\lambda}^*}{d\gn^*}=\frac{2v^{*^2}}{\left(\gn^*+g_l\right)^2}\left(\gn^*-2g_l\right),
    $$
    which is negative in $[g_l,2g_l)$, positive in $(2g_l,+\infty)$, and zero for $\gn^*=2g_l$. Therefore, $\tilde{\lambda}^*$ is strictly decreasing in $[g_l,2g_l)$, strictly increasing in $(2g_l,+\infty)$, and
    $$\min_{\gn^*\in[g_l,+\infty)}\tilde{\lambda}^*=\tilde{\lambda}(2g_l)=-\frac{E^2}{27}.
    $$
    Let $h(\gn^*)=\lambda(\gn^*)+b$ and $b\in\left(0,\frac{E^2}{27}\right)$ 
    $$
        h(g_l)=\tilde{\lambda}(g_l)+b=b>0, \quad  h(2g_l)=\tilde{\lambda}(2g_l)+b<0.
    $$
    Since $h(\gn^*)$ is continuous on $\gn^*$, by the IVT there exists $\xi_b\in(g_l,2g_l)$ such that $h(\xi_b)=0$, or, equivalently, $h(g_l+m_b)=0$ for $m_b=\xi_b-g_l$. Since $h$ is strictly decreasing on $(g_l,2g_l)$, $\xi_b$ is the unique intersection of $h$ with the zero axis and $h(\gn^*)>0$ for $\gn^*\in(g_l, g_l+m_b)$.
    
    Similarly, for $\gn^*\in(2g_l,+\infty)$, 
    $$
        h(2g_l)=\tilde{\lambda}(2g_l)+b<0,\quad \lim_{\gn^*\rightarrow+\infty}\tilde{\lambda}(\gn^*)+b=b>0.
    $$
    Since $h(\gn^*)$ is continuous, strictly increasing in $(2g_l,+\infty)$, there exists unique $\Xi_b\in(2g_l,+\infty)$ that $h(\Xi_b)=0$, or, equivalently, $h(g_l+M_b)=0$ for $M_b=\Xi_b-g_l$. When $b\in\left(0,\frac{E^2}{27}\right)$, combining the two results leads to
    \begin{equation}\label{eq:sign-tilde-lambda+b}
        h(\gn^*)=\lambda(\gn^*)+b:\begin{cases}
          \geq 0, &\quad \forall \gn^*\in [g_l,g_l+m_b) \\
          \leq 0, &\quad \forall \gn^*\in \left[g_l+m_b,g_l+M_b\right) \\
          \geq 0, &\quad \forall \gn^*\in \left[g_l+M_b,+\infty\right)
        \end{cases}
    \end{equation}
    Moreover,
    \begin{align*}
        \eqref{eq:c(gt)}\Rightarrow& c^*=\left(\frac{ v^{* 2}}{\gn+
		g_l}\left(g_l-\gn^*\right)\right)^2 > 0, \\
		\eqref{eq:tilde-lambda}\Rightarrow& \tilde{\lambda}^*+v^{*^2}+b=2v^{*^2}\left(1-\frac{\gn^*}{\gn^*+g_l}\right)+b >0.
	\end{align*}
	By substituting the above into \eqref{eq:r-h-b1-cond} leads to $f(\gn^*)<0$ for $g_l+m_b\leq \gn^* \leq g_l+M_b$ and subsequently $b_1(g^*)<0$. By the Ruth-Hurwitz criterion, when $b\in\left(0,\frac{E^2}{27}\right)$, then $\left\{\left(g^*,\,\,\phi^*,\,\,\hat{g}^*\right)\in\mathcal{E}_s:g_l+m_b\leq \gn^*\leq g_l+M_b\right\}$ is unstable.
	
	(2) Let $b_1,b_2\in\left(0,\frac{E^2}{27}\right)$, $b_1<b_2$. By applying the IVT on $(g_l,2g_l)$ we have shown that there exist unique $0<m_{b_1},m_{b_2}\leq g_l$ such that
	\begin{align*}
        &\tilde{\lambda}(g_l+m_{b_1})+b_1=\tilde{\lambda}(g_l+m_{b_2})+b_2=0 \\
        \Rightarrow& \tilde{\lambda}(g_l+m_{b_1}) - \tilde{\lambda}(g_l+m_{b_2}) = b_2 - b_1 >0 \\ 
        \Rightarrow& \tilde{\lambda}(g_l+m_{b_1}) > \tilde{\lambda}(g_l+m_{b_2}). 
	\end{align*}
	Since $\tilde{\lambda}$ is strictly decreasing in $(g_l,2g_l)$, then 
	\begin{equation}\label{eq:mb1-mb2}
	    g_l+m_{b_1} < g_l+m_{b_2} \Rightarrow m_{b_1} < m_{b_2}.
	\end{equation}
    Moreover, by applying the IVT on $(2g_l,+\infty)$ we have shown that there exist $g_l<M_{b_1},M_{b_2}<+\infty$ such that
    \begin{align*}
        &\tilde{\lambda}(g_l+M_{b_1})+b_1=\tilde{\lambda}(g_l+M_{b_2})+b_2=0 \\
        \Rightarrow& \tilde{\lambda}(g_l+M_{b_1})-\tilde{\lambda}(g_l+M_{b_2})=b_2 - b_1 > 0 \\
        \Rightarrow& \tilde{\lambda}(g_l+M_{b_1}) > \tilde{\lambda}(g_l+M_{b_2})
	\end{align*}
	Since $\tilde{\lambda}$ is strictly increasing in $(2g_l,+\infty)$, then
	\begin{equation}\label{eq:Mb1-Mb2}
	    g_l+M_{b_1} > g_l+M_{b_2} \Rightarrow M_{b_1} > M_{b_2}.
	\end{equation}\end{proof}

\end{document}